\documentclass[draftclsfinal,onecolumn]{IEEEtran}
\pdfoutput=1
\usepackage{url}
\usepackage{cite}
\usepackage{latexsym}
\newtheorem{proposition}{Proposition}
\newtheorem{corollary}{Corollary}

\newtheorem{remark}{Remark}
\newtheorem{theorem}{Theorem}
\newtheorem{definition}{Definition}
\usepackage{amsmath}
\usepackage{amstext,amsfonts,epsf,epsfig,pifont}
\usepackage{graphics,bm,amsmath}
\newlength{\intwidth}

\def\XXint#1#2#3{{\setbox0=\hbox{$#1{#2#3}{\int}$}
\vcenter{\hbox{$#2#3$}}\kern-.5\wd0}}

\usepackage{amsmath}
\usepackage{amsfonts,mathrsfs}
\usepackage{enumerate}
\usepackage[bottom]{footmisc}

\usepackage{float}
\usepackage{epsfig}
\usepackage{bm}
\newcommand{\cM}{\ensuremath{{\mathcal{M}}}}
\newcommand{\cB}{\ensuremath{{\mathcal{B}}}}

\newcommand{\bpsi}{\ensuremath{{\bm{\psi}}}}
\newcommand{\cN}{\ensuremath{{\mathcal{L}}}}
\newcommand{\cQ}{\ensuremath{{\mathcal{Q}}}}

\newcommand{\E}{\ensuremath{{\mathrm{E}}}}
\newcommand{\var}{\ensuremath{{\mathrm{var}}}}
\newcommand{\cov}{\ensuremath{{\mathrm{cov}}}}
\newcommand{\rel}{\ensuremath{{\mathrm{rel}}}}

\title{Learning the Ambiguity Surface}
\author{Sofia~C.~Olhede
\thanks{The work of S.~C.~Olhede was supported by the EPSRC.}
\thanks{S.~C.~Olhede is with the Department of Statistical Science, University College London, Gower Street,
London WC1 E6BT, UK.}}

\begin{document}

\maketitle

\begin{abstract}
This paper introduces the class of ambiguity sparse processes, containing subsets of popular nonstationary time series such as locally stationary, cyclostationary and uniformly modulated processes. The class also contains aggregations of the aforementioned processes. Ambiguity sparse processes are defined for a fixed sampling regime, in terms of a given number of sample points and a fixed sampling period. The framework naturally allows us to treat heterogeneously nonstationary processes, and to develop methodology for processes that have growing but controlled complexity with increasing sample sizes and shrinking sampling periods. 

Expressions for the moments of the sample ambiguity function are derived for ambiguity sparse processes. These properties inspire an Empirical Bayes shrinkage estimation procedure. The representation of the covariance structure of the process in terms of a time-frequency representation is separated from the estimation of these second order properties. The estimated ambiguity function is converted into an estimate of the time-varying moments of the process, and from these moments, any bilinear representation can be calculated with reduced estimation risk. Any of these representations can be used to understand the time-varying spectral content of the signal. The choice of representation is discussed.  Parameters of the shrinkage procedure quantify the performance of the proposed estimation.
\end{abstract}
\section{Introduction \label{sec:intro}}
Assumptions on stationarity dominate time series analysis, and are a necessary requirement for most of traditional time series methodology to work. Unfortunately many practical inference problems violate these assumptions. For this reason the last fifty years of time series analysis have seen considerable developments in generalizing methodology to encompass processes violating stationarity constraints, see for example the class of semi-stationary processes introduced by
\cite{Priestley65}, the class of locally stationary processes due to \cite{Silverman1957}, methods of estimating local spectral summaries
in a nonparametric framework introduced by 
\cite{Dahlhaus1997,Dahlhaus2000} and \cite{Ombao2001}.
Time domain notions of local stationarity are also discussed by 
{\em e.g.} \cite{Mallat1998} {\em etc}; and we mention wavelet based models, see {\em e.g.} \cite{Nason}. These developments have significantly advanced the tool-kit available for data analysis and modelling, but when using them we are still strongly limited in terms of permitted ways of modelling the covariance of a process so that it can be estimated stably. Also while within a single framework there is a representation of the covariance function as a ``time-varying spectrum,'' this is most usually poorly defined between  model classes and does not yield a model independent definition.

In tandem to these developments in statistics, in signal processing the theory of bilinear representations of covariance structures has reached maturity, see {\em e.g.} Flandrin \cite[Ch.~3]{Flandrin99} or \cite{cohen}. Modern understanding of time-varying spectral representations is that no particular family of representation is uniformly optimal across different model-classes in terms of representing structure, and 
there has been less focus on
producing estimators. 
For example two broad classes of nonstationary time series are those of Harmonizable processes \cite{Loeve1963} and Karhunen processes \cite{Karhunen1947}; however the practical utility of these classes is limited, as they have no general inference methods associated with them.
Of particular note in modern development is the class of {\em underspread processes} introduced by Matz, Hlawatsch and Kozek (see {\em e.g.} \cite{Matz1997,Matz2006}), with associated inference methods, {\em e.g.} \cite{Jachan}. Underspread processes are those which have been sampled sufficiently rapidly to be considered smoothly varying with the given sampling, and these are strongly linked to the class of semi-stationary processes, that are constrained in terms of properties of a Fourier transform of the local spectrum called the Lo\`eve spectrum, see {\em e.g.} \cite{Flandrin99}.
All of the existing inference methods for these types of processes are based on a notion of local smoothing, and removing medium and high frequency structure. This violates how modern data is normally acquired -- in contrast to an underspread process we expect the complexity of the data to {\em grow} as we {\em increase} the sampling size $N$ and {\em decrease} the sampling period $\Delta t$, even if we do expect some underlying simplicity in the representation of this process.
To model this kind of process
we introduce a new class of nonstationary processes called {\em ambiguity sparse processes}, whose Ambiguity Functions (AF) are mainly supported over a set of regions. The AF is a Fourier Transform of the autocovariance sequence transforming in the global time index rather than in the time-lag. Processes that are ambiguity sparse, if sampled more finely than the available data, {\em are} also underspread; however we expect data to grow in complexity with the sampling. 

After defining the AF from the time-varying covariance of the process, we derive the statistical properties of its sample version for an ambiguity sparse process, which inspires an Empirical Bayes' method of estimation. We discuss the properties of this estimation method, and describe how the estimate of the AF can be converted into an estimate of the basic object of the representation of a nonstationary process, namely the covariance sequence of the observed process. This covariance sequence in turn can be represented in terms of any local spectral estimator in the bilinear class, see {\em e.g.} Flandrin \cite[Ch.~3]{Flandrin99}. We discuss the interpretation of our estimators, as ``smoothing'' the sample covariance sequence adaptively. The performance of the estimation can be quantified in terms of parameters in the Empirical Bayes procedure. An additional benefit of ambiguity sparse processes is that it is very easy to characterise when sums of the processes, or indeed linear filtered versions of such processes, remain in the class of ambiguity sparse processes. 

We illustrate the properties of the proposed estimation procedures on a number of simulated and real data examples. We discuss different choices of time-varying spectral representation, and two different methods of regularizing the estimation procedure to ensure that the estimated autocovariance sequence is positive definite, and thus corresponds to a valid estimator of covariance.

\section{The Four Faces of Time-Frequency}\label{fourfaces}
\subsection{Time-Time Representation}
We assume a single real-valued time series $\{X_n\}$ is under analysis. This series is assumed to be mean-zero, or it is assumed that the true mean has already been estimated and removed. 
Furthermore we assume $\{X_n\}$ is a Gaussian process. Because we have made few assumptions on the smoothness or structure of the covariance sequence of $\{X_n\}$, many distributions could be well approximated by this mixture of Gaussians. $\{X_n\}$ are the samples of a continuous time process $\{X(t)\}$ at time points
$
t_n=n \Delta t,\quad {\mathrm{with}}\quad n=0,\dots, N -1,
$
and $\Delta t>0$ is the {\em sampling period}. Because $\{X_n\}$ is a Gaussian process and we know the first moment is zero to make inferences we need to model the second order structure or 
\begin{equation}
\cM_{\tau}(t_n)=\E\left\{X_n X_{n-\tau}\right\},\quad \tau=-N+n+1,\dots n,
\,n=0,\dots,N-1.\end{equation}
Writing down this definition is trivial, however finding processes for which methods can be designed such that
$\{\cM_{\tau}(t_n)\}$ can be estimated consistently from the available data is a substantially harder problem. If a process is stationary then $\cM_{\tau}(t_n)=\widetilde{\cM}_{\tau}$ for all $t_n$ and additionally the sequence $\{\widetilde{\cM}_{\tau}\}$ has elements that are finite for all $\tau$, where the elements can be estimated stably either using method of moments, or a parametric approximation to the full covariance, such as using a truncated version of the Wold decomposition, {\em e.g.} a moving average model. 
We refer to $t_n\in{\mathbb{R}}$, with $n \in{\mathbb{Z}}$ as the {\em global} time index, and $\tau\in{\mathbb{Z}}$ as the {\em local} time index, or {\em time-lag}. 

An important modelling tool for stationary time series is the spectral density function (sdf) defined when the integrated spectrum is absolutely continuous \cite{Adler} if $\{X_n=X(t_n)\}$ has been sampled sufficiently finely in time by \cite[p.~196]{PercivalWalden1993},
\begin{eqnarray}
\widetilde{S}(f)=\Delta t\sum_{\tau=-\infty}^{\infty} \widetilde{\cM}_{\tau} e^{-2i\pi f \tau \Delta t},\quad {\mathrm{for}}\; -\frac{1}{2\Delta t}\le f <\frac{1}{2\Delta t},
\end{eqnarray}
and substantial efforts have also gone into defining transforms of $\{\cM_{\tau}(t_n)\}_\tau$ to an {\em interpretable} time-varying spectral representation, see {\em e.g.} \cite{loynes,Priestley65}, \cite[Ch.~6]{cohen} or \cite[Ch.~2.3, 2.4, 3.1, 3.3]{Flandrin99}. Issues arise on two different counts; firstly it is very hard to define a suitable time-varying spectral density definition from a {\em known} and given deterministic $\cM_{\tau}(t_n)$, secondly designing estimation methods with suitable properties of this object is even harder (see some early discussion in \cite{martinflandrin1985}).
 
It is well known that defining time-varying spectra from the covariance structure of real-valued signals is problematic. Interference from negative frequencies can cause non-interpretable terms to appear in any bilinear representation \cite{Jeong92}. Thus to improve the mathematical properties of representations of the second order structure it is common practice to calculate the representation of the analytic signal of $\{X_n\}$. 
We define the analytic signal by (with the assumption that $\Delta t$ has been sampled finely to make out the  important characteristics of the series)
\begin{eqnarray}
Z_n &=& 2\int_0^{\infty} dZ_X(f)e^{2i\pi ft_n}\;df=X_n+i Y_n,\quad Y_n={\cal H}\left\{X_n\right\},
\label{complex}
\end{eqnarray}
where $Y_n$ is defined by this equation, and ${\cal H}\left\{\cdot\right\}$ is the Hilbert transform. For a finite sample with fixed $\Delta t$ and $N$
we approximate ${\cal H}\left\{\cdot\right\}$ using the discrete Hilbert transform. 
$\{Z_n\}$, or the {\em analytic signal}, can exhibit no interference between negative and positive frequencies.

It may be the case that the positive frequencies are still correlated with their conjugate, see \cite{Schreier}, and to complete our second order description we do not only need to model and estimate the autocovariance of $\{Z_n\}$ defined by (recalling $\E\left\{Z_n\right\}=0$)
\begin{eqnarray}
\label{moments_in_time}
M_{\tau}(t_n)=\E\left\{Z_n Z_{n-\tau}^*\right\},
\end{eqnarray}
but also we need to model $R_{\tau}(t_n)=\E\left\{Z_n Z_{n-\tau}\right\}$ (the so-called {\em relation sequence} \cite[p.~55]{Schreier}). Discussing improper stochastic processes ({\em e.g.} those processes for which $R_\tau(t_n)\neq 0$) is outside the scope of this paper and we note in passing that the methods we shall propose for the covariance sequence can also be straightforwardly extended to estimating the relation sequence.
We do have to pay a price for analyzing $\{Z_n\}$ rather than $\{X_n\}$: $\{Y_n\}$ as defined in \eqref{complex}
is in general a ``smeared out'' version of $\{X_n\}$, and so $\cM_{\tau}(t_n)\neq M_{\tau}(t_n)$. 

\subsection{Time-Frequency Representation}
A fundamental problem with defining a time-varying spectrum from $\{M_{\tau}(t_n)\}$ is instability \cite{Matz1997}.
We could in theory shift the sequence in time so that we estimate 
$M_{\tau}^{(\alpha)}(t_n)\equiv M_{\tau}\left(t_{n+\left(\frac{1}{2}-\alpha\right) \tau}\right)$, and it would be equally reasonable to call this object the ``local time-moment'' of the process. \footnote{There is a problem in calculating this object if $\left(\pm\frac{1}{2}-\alpha\right)\tau \notin {\mathbb{Z}}$, that can either be solved using interpolation or by only using the moments corresponding to integer valued times.} 
Furthermore in general $M_{\tau}^{(\alpha_1)}(t_n)\neq M_{\tau}^{(\alpha_2)}(t_n)$ for $\alpha_1\neq \alpha_2$.
We could also subsequently calculate Fourier transforms in $\tau$ for any given choice of $\alpha\in\left[-\frac{1}{2},\frac{1}{2}\right],$ see \cite{Matz1997}, to describe the time-varying spectral content of the process,
with the global time shifted to
$t_n^{\alpha}=\left(n+\left(\frac{1}{2}-\alpha\right)\tau\right)\Delta t.$ We could also define a local spectral representation by {\em any} shift $\alpha$ centering the moments in time
\begin{equation}
\label{moments1}
{\mathcal S}^{\alpha}(t_n,f)=\Delta t\sum_{\tau=-\infty}^{\infty} 
M_{\tau}\left(t_n^{\alpha}\right)e^{-2i\pi f \tau \Delta t},
\end{equation}
where
$f$ is referred to as the {\em global} frequency. Each choice of $\alpha$ de facto leads to a different choice of a time-varying spectral representation.
The most popular choice is $\alpha=0$ \cite{cohen}.
We shall focus on estimating ${\cal S}^0(t_n,f)$, from smoothing an empirical estimator of $M_{\tau}(t)$, but also will compare this representation with other choices, such as $\alpha=1/2$, corresponding to the {\em Rihaczek} distribution \cite{Rihaczek}. 
A larger class of time-frequency distributions is formed by additionally allowing a convolution of $ {\mathcal S}^0(t_n,f)$ in time {\em and} frequency, see \cite{cohen}. We define the Short-Time Fourier Transform (STFT) as (for some given window function $h(\cdot)$), by
$
J^{(z,h)}(t,f)=\sqrt{\Delta t}\sum_{s=-\infty}^{\infty} h(s-t) Z_s e^{-2i\pi fs \Delta t}.
$
In statistics inference is often based on the STFT ({\em e.g.} \cite{Dahlhaus2000}), and an important characteristic is its variance:
\begin{alignat}{1}
\E\left\{\left|J^{(z,h)}(t,f)\right|^2\right\}
&=\Delta t\sum_{t_1=-\infty}^{\infty}\sum_{\tau=-\infty}^{\infty} \varpi_\tau(t_1-t) M_{\tau}(t_1) e^{-2i\pi f\tau\Delta t},\;
\varpi_\tau(t)=h(t)h(t-\tau).
\label{SFT_rep}
\end{alignat}
The expected magnitude of the STFT can also be represented as a FT of the moments $\{M_\tau(t)\}$ of $\{Z_t\}$.
In general the window functions $h(\cdot)$ can be replaced by an arbitrary nonseparable kernel $\varpi_\tau(t)$ for greater flexibility, catering to variable smoothness of $M_\tau(t)$, thus allowing for better frequency resolution when such is possible. The function $M_\tau(t)$ has a different degree of smoothness in $\tau$ and $t$ and so using Eqn.~\eqref{SFT_rep} to  threshold adaptively is better than the non-adaptive choice of $\varpi_\tau(t)=h(t)h(t-\tau)$. The variance of $\left|J^{(z,h)}(t,f)\right|^2$
is {\em large} and additional smoothing is required to make this a good estimator. Clearly whatever is done subsequently to calculating $\left|J^{(z,h)}(t,f)\right|^2$
depends strongly on the choice of the function $h(\cdot)$.

\subsection{Frequency-Time Representation}
We have defined a Fourier transform of $M_{\tau}\left(t_n\right)$ as the time-varying spectrum. The Fourier transform of the moment sequence in global time $t$ instead of in lag $\tau$ defines the function
\begin{equation}
\label{AFdef}
A_{\tau}(\nu)=\Delta t \sum_{n=-\infty}^{\infty} M_{\tau}(t_n)\;e^{-2i\pi \nu t_n}.
\end{equation}
$A_{\tau}(\nu)$ is the {\em ambiguity function} (AF) of $\{X_n\}$, see {\em e.g.} \cite{Blahut1991ed,Flandrin99}.
$A_{\tau}^{\alpha}(\nu)$ can also be defined for any arbitrary choice of $t_n^{\alpha}$, rather than centering the global time at $t_n$, as in Eqn \eqref{AFdef},
however this will only have an effect as a phase-shift term, and will not change the support of the AF, as this depends only on the magnitude. 
The AF has been used to define underspread processes \cite{Matz2006}, e.g. those whose AFs are limited to a set of frequencies near the origin, or $(\nu,\tau)=(0,0)$.  
Very few useful processes are strictly underspread. For this reason Hlawatsch and Matz introduced {\em nearly underspread} processes which only require a suitable decay of the magnitude of the AF from the origin, {\em e.g.} that $|A_\tau(\nu)|$ decays from the origin. Such processes will inevitably fit into existing theory in terms of inference, see e.g. \cite{Dahlhaus2000}, with the associated notion of asymptotics.

For completeness we define the dual-frequency spectrum $S(\nu,f)$ as a FT of the AF
$S(\nu,f)=\Delta t\sum_{\tau=-\infty}^{\infty} A_{\tau}(\nu)e^{-2i\pi f\tau \Delta t }.$
As $\{X_n\}$ is assumed harmonizable
this quantity has an interpretation by using the spectral representation of $\{Z_n\}$.
We let the Lo\`eve spectrum ($S_L(f_1,f_2)$) be given by\\
$\cov\{d{ Z}(f_1),d{ Z}(f_2)\}=S_L(f_1,f_2)\,df_1df_2.$
We find that the dual-frequency spectrum is a shifted version of the Lo\`eve spectrum as
$S(\nu,f)=
S_{L}(f+\nu,f).$
The point of this shift is to align the stationary spectrum to concur structurally with $S(0,f)$.
We can now describe the second order structure of the nonstationary process by any of these four different functions, namely the time-varying autocovariance $M_{\tau}(t_n),$
the dual-frequency spectrum $S(\nu,f),$ the time-varying spectrum ${\mathcal S}(t,f)$ or the AF $A_{\tau}(\nu)$. 
Because of the uniqueness of Fourier transforms one might argue that any of these descriptions is {\em equally} useful, important and indeed complete.

We shall introduce a new class of processes which if they did not grow in complexity with $N$ and $\Delta t$ and if you observed them over sufficiently large time-windows and at sufficiently small $\Delta t$ would be underspread.
As we collect more data, we generally observe processes of growing complexity -- hence we do  wish to define processes whose complexity nominally may {\em increase} with {\em decreasing} $\Delta t$ and {\em increasing N}. We shall estimate their second order properties for a given (small) $\Delta t$ and (large) $N$. We never anticipate to have all the energy of the process living at very low values of $\nu$, and we wish to reduce variance by instead zeroing out large (not necessarily connected) regions of the ambiguity plane. 
We therefore prefer to introduce these processes for a {\em fixed sampling scheme}
as centred at a few time-frequency points and exhibit decay away from these points. 
\begin{definition}{Ambiguity Sparse Process\label{ambsparse}}\\
A second order real-valued time series $\{X_n=X(n\Delta t)\}$ is denoted {\em Ambiguity Sparse} at sampling $N,\Delta t$ if its AF can be represented for $K\in {\mathbb N}$  in the form 
\begin{equation}
A_{\tau}(\nu)= \sum_{k=1}^K A_{\tau}^{(k)}(\nu),\,
A_{\tau}^{(k)}(\nu)= {\cal B}^{(k)}(\nu,\tau/N)
\left[\Delta t^2\left(\nu-\nu_0^{(k)}\right)^2+\left(\tau-\tau_0^{(k)}\right)^2/N^2\right]^{-\delta^{(k)}}
\label{AFmodel},
\end{equation}
with ${\cal B}^{(k)}(\nu,u)$ a smooth function near $(\nu_0^{(k)},\tau_0^{(k)})$, taking a non-zero value at this point, $\delta^{(k)}>0$. 
\end{definition}
To cope with more anisotropic structure, such as chirping components, Eqn \eqref{AFmodel} can be rotated and dilated to convert the circular contours of \eqref{AFmodel} into elliptical contours with a given orientation, and  all (subsequent) proofs would change {\em mutatis mutandis} in the rotated frame.
If we define
$
\mathbf{y}_k={\mathrm{diag}}\left(
a_1^{(k)} \; a_2^{(k)}\right)
{\mathbf{R}}_{-\theta}\left(
\Delta t \left(\nu-\nu_0^{(k)}\right)\;
\left(\tau-\tau_0^{(k)}\right)/N
\right)^T ,
$
with ${\mathbf{R}}_{-\theta}=\{(\cos(\theta) -\sin(\theta),(\sin(\theta) \cos(\theta))\}$ this would correspond to a model of 
\begin{equation}
A_{\tau}^{(k)}(\nu)\approx {\cal B}^{(k)}(\nu,\tau/N)
\left[\mathbf{y}_k^T\mathbf{y}_k\right]^{-\delta^{(k)}}
\label{AFmodel2}.
\end{equation}
Such a model would naturally treat stochastic chirping signals, see for example \cite{Martin81}. This allows nonstationary structure not necessarily aligned with the usual time-frequency axes.

\begin{remark}{The Sum of Two Ambiguity Sparse Processes}\\
The sum of two ambiguity sparse processes at a fixed sampling is ambiguity sparse, as long as all the spreading kernels $\{\Psi_n(\nu)\}$ are limited in support. The time-varying sequence $\{\psi_n(t)\}$ (and their Fourier transform $\{\Psi_n(\nu)\}$) is defined so that
{\small\begin{equation}
Y_n=\Delta t \sum_{h=-\infty}^\infty \psi_h^\ast(t_n)X_{n-h}+Y^{\perp}_n,
\;
\Delta t\sum_{h=-\infty}^{\infty}\psi_h^\ast(t_n)M_{\tau-h}(t-h)=\cov\left\{Y_n,X_{n-\tau} \right\}=M^{(Y,X)}_\tau(t_n), \end{equation}}
so that
$\cov\left\{Y^{\perp}_t,X_{t-h}\right\}=0$ for all $h\in{\mathbb{Z}}$. Then
with $A^{(X)}_\tau(\nu)$ as the AF of a single process $\{X_t\}$ and $A^{(Y,X)}_\tau(\nu)$
as the FT of $M^{(Y,X)}_\tau(t_n)$ in $t_n$ (the cross-Ambiguity Function) we have
\begin{eqnarray}
A_\tau^{(X+Y)}(t)&=&A_\tau^{(X)}(t)+A_\tau^{(Y)}(t)+A_\tau^{(X,Y)}(t)+A_\tau^{(Y,X)}(t),
\\
A_\tau^{(Y,X)}(\nu)
&=&\Delta t \sum_{h=-\infty}^\infty \int_{-\infty}^\infty \Psi_h^\ast\left(-\left(\nu'+\nu\right)\right) A_{\tau-h}(\nu')e^{2i\pi h \nu'}\;d\nu'.
\end{eqnarray}
We therefore get a countably infinite number of kernels $\{\Psi_n(\nu)\}$ spreading $A_{\tau}(\nu)$, when constructing $A_\tau^{(Y,X)}(\nu)$ (and similarly $A_\tau^{(X,Y)}(\nu)$.) If the support of {\em each} of $\{\Psi_h(\nu)\}$ is limited in $\nu$ and $h$ then $A_\tau^{(Y,X)}(\nu)$ and $A_\tau^{(X,Y)}(\nu)$ are still sparse. $A_\tau^{(X)}(\nu)$ is sparse by assumption; so is $A_\tau^{(Y)}(\nu)$.
Compare with the theory for semi-stationary processes, see 
\cite{Tong1973,Tong1974}, which is complicated because a single oscillatory family must be used for both the processes $\{X_t\}$ and $\{Y_t\}$ and it is hard to untangle potential correlation between processes.
\end{remark}
\begin{remark}
The notion of ambiguity sparse can directly be related for models of the smoothness of $M_\tau(t)$. Assume that $M_\tau(t),$ uniformly in $\tau,$ is Sobolev order $k$. This implies that
$
\left|A_\tau(\nu)\right|\le {\mathscr A} \nu^{-(k+1)},$
for sufficiently large $\nu$ and positive constant ${\mathscr A}$. This class fits into that of ambiguity sparse processes, but also into the class of underspread processes. 
\end{remark}


\section{Method of Moments Estimation}\label{methmoments}
We define a naive estimate of the analytic autocovariance function by
\begin{eqnarray}
\widehat{M}_{\tau}(t_n)=\left\{\begin{array}{lcr}
Z_n Z_{n-\tau}^\ast & {\mathrm{if}} & n=0,\dots -(N-n-1),\,\tau=0,\dots n-1 \\
0 &  & {\mathrm{otherwise}}
\end{array} \right. .
\end{eqnarray}
It is natural to start to define a time-varying spectrum starting from $M_\tau(t_n)$. We may form different time-varying spectral estimators by filtering any such choice of time-varying spectrum that is calculated from $\widehat{M}_\tau(t_n)$. 
$\widehat{M}_{\tau}(t_n)$ is unbiased when estimating the autocovariance of the analytic signal. Unfortunately its variance is very large. For example as $\{X_n\}$ is a Gaussian process then by Isserlis' theorem (see \cite{Isserlis1918}) we can note that 
\begin{eqnarray}
\label{raw_var}
\var\left\{\widehat{M}_{\tau}(t_n)\right\}&=&M_{0}(t_n)M_{0}(t_{n-\tau})+R_{\tau}(t_n)R_{-\tau}^\ast(t_{n-\tau}).
\end{eqnarray}
With our assumption of propriety the second term in  Eqn \eqref{raw_var} is identically zero, but despite this,
the estimator is  to all intents and purposes useless, as it is far too variable. The usual approach to improving a variable but unbiased estimator is via some form of smoothing. We define the smoothed estimator (for some chosen kernel function $\omega_\tau(t)$) as
\begin{eqnarray}
\label{kernel:est}
\overline{M}_{\tau}(t_n)=\Delta t\sum_{j=-\infty}^{\infty} \omega_{\tau}(t_{j-n})\widehat{M}_{\tau}(t_j).
\end{eqnarray}
Ideally $\omega_{\tau}(t_n)$ would be a variable bandwidth kernel that is naturally adapted to the smoothness of ${M}_{\tau}(t)$ in $t$ for each given value of $\tau$ ({\em cf} Eqn \eqref{SFT_rep}). Some examples of possible kernels for examples given in this paper are shown in Figure \ref{fig9}. These reinforce different structure, such as implementing local smoothing but keeping dependence across all lags (subplot (b) and (d)), reinforcing cyclostationary structure (subplots (a) \& (c)), and truncating dependence (subplot (a)). 

\begin{figure}[!htbp]
\begin{center}
\begin{minipage}[]{0.40\textwidth}
\centering
\includegraphics[width=\textwidth]{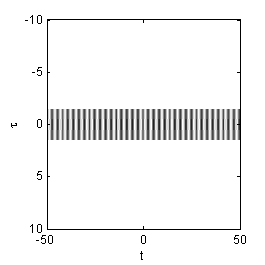}
(a)
\end{minipage}
\begin{minipage}[]{0.40\textwidth}
\centering
\includegraphics[width=\textwidth]{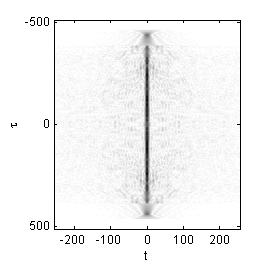}
(b)
\end{minipage}\\
\begin{minipage}[]{0.40\textwidth}
\centering
\includegraphics[width=\textwidth]{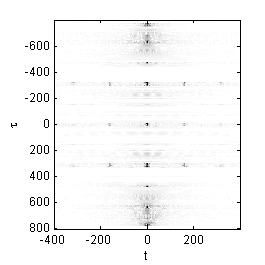}
(c)
\end{minipage}
\begin{minipage}[]{0.40\textwidth}
\centering
\includegraphics[width=\textwidth]{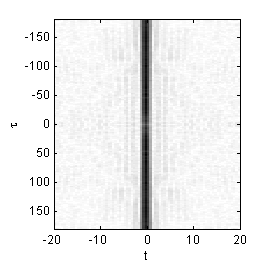}
(d)
\end{minipage}

\end{center}
\caption{The magnitude of the initial ``smoothing'' window used to smooth the empirical autocovariance function for the aggregation signal (a), for the simulated blood flow signal (b), for the bat signal (c), and for the oceanographic ``meddy'' signal (d).}
\label{fig9}
\end{figure}
If we Fourier transform $\overline{M}_{\tau}(t_n)$, then we obtain an estimator of the AF of $\{X_n\}$. We define the EMpirical AF (EMAF) by a DFT of $\widehat{M}_{\tau}(t_n)$
\begin{eqnarray}
\label{eqn:EMAF}
\widehat{A}_{\tau}(\nu)=\Delta t\sum_{n=0}^{N-1} \widehat{M}_{\tau}(t_n)e^{-i2\pi \nu t_n}.
\end{eqnarray}
$\widehat{A}_{\tau}(\nu)$ suffers from {\em exactly} the same problems as $\widehat{M}_{\tau}(t_n)$, {\em e.g.} a large variance.
We note that
\begin{eqnarray}
\overline{A}_{\tau}(\nu)&\equiv&\Delta t \sum_{n=0}^{N-1} \overline{M}_{\tau}(t_n)e^{-i2\pi \nu t_n}=\Omega_{\tau}(\nu)\widehat{A}_{\tau}(\nu),\quad \Omega_{\tau}(\nu)=\Delta t \sum_{n=0}^{N-1} \omega_{\tau}(t_{n})e^{-i2\pi \nu t_n}.
\label{smooth}
\end{eqnarray}
Thus our estimator of the ambiguity estimator is a multiplication of the EMAF with the FT of the kernel $\omega_\tau(t_n)$, and so if $|\Omega_{\tau}(\nu) |\le 1,$ $\overline{A}_{\tau}(\nu)$ is a {\em shrinkage} estimator of $A_{\tau}(\nu)$. 
We must decide how to chose an appropriate kernel, or the degree of shrinkage at each lag and nonstationary frequency. This must be done considering the stochastic properties of the EMAF $\{\widehat{A}_\tau(\nu)\}$. Traditional estimators would smooth $\{\widehat{M}_\tau(t)\}$ in $t$ uniformly in $\tau$ to reduce mean square error -- this corresponding to shrinking contributions non-adaptively at larger values of $|\nu|$, see {\em e.g.} \cite[p.~4372]{Jachan}. We wish to retain more flexibility, and so do not wish to put in strong structural assumptions in our treatment of $\widehat{A}_\tau(\nu)$. Shrinkage must therefore be implemented adaptively.
With the model of Definition \ref{ambsparse} we shall determine the statistical properties of $\widehat{A}_{\tau}(\nu)$, to construct good shrinkage estimators of $A_{\tau}(\nu)$.
\begin{theorem}{Concentration of Ambiguity Surface\label{ConcAmbSurf}}\\
$\{X_n\}$ is assumed to be an ambiguity sparse process  satisfying Eqn \eqref{AFmodel}, and $T(\tau)=(N-|\tau|)\Delta t$.
\begin{enumerate}
\item If $\delta^{(k)}\in\left[\frac{1}{4},\frac{3}{4}\right]$ then
$
\E\left\{\widehat{A}_\tau(\nu)\right\}
\approx \sum_{k=1}^K A_\tau^{(k)}(\nu){\cal J}_1(\delta^{(k)};\{T(\tau)(\nu-\nu_0^{(k)}),\tau-\tau_0^{(k)}\}),
$
where ${\cal J}_1(\delta^{(k)};\{T(\tau)(\nu-\nu_0^{(k)}),\tau-\tau_0^{(k)}\})\sim 1/(2\Delta t)$  as long as $(\nu,\tau)$ is near one of the $K$ contributions; otherwise\footnote{The factor of $1/(2\Delta t)$ is expected as the bandwidth of the analytic signal. We only need to recover the AF as a function upto proportionality and can easily multiply our estimate by $2\Delta t$.}  ${\cal J}_1(\delta^{(k)};\{T(\tau)(\nu-\nu_0^{(k)}),\tau-\tau_0^{(k)}\})\sim 1/(2\Delta t)-|\nu|$. Furthermore if $(\nu,\tau)\in {\cal D}_k(L)$  where
we fix $L>0$, and let
\[{\cal D}_k(L)=\left\{(\nu,\tau):
\sqrt{N^2\Delta t^2(\nu-\nu_0^{(k)})^2+(\tau-\tau_0^{(k)})^2}\le L\right\},\;k=1,\dots,K,\]
then
\[{\cal O}( A_\tau(\nu){\cal J}_1(\delta^{(k)};\{T(\tau)(\nu-\nu_0^{(k)}),\tau-\tau_0^{(k)}\})=(N-|\tau|)^{2\delta^{(k)}}/(2\Delta t)=
{\cal N}_\tau(\nu;N,\Delta t)
.\]

\item If Eqn \eqref{AFmodel} holds with $\delta^{(k)}\in\left(\frac{1}{4},1\right)$ the variance of $\widehat{A}_\tau(\nu)$ takes the form:
\begin{eqnarray} \nonumber
\var\left\{\widehat{A}_{\tau}(\nu)\right\}
&=& {\cal K}_{\tau}(\nu;N,\Delta t){\cal V}_{\tau}(\nu),
\label{by2}\end{eqnarray}
${\cal K}_{\tau}(\nu;N,\Delta t)=(N-|\tau|)^{4\max_k\delta^{(k)}-1}\frac{1}{\Delta t}\left(\frac{1}{2\Delta t }-|\nu|\right)$
with the normalized ambiguity variance ${\cal V}_{\tau}(\nu)$ defined implicitly. 
\item If $(\nu,\tau)$ is not near any of the centre points then the relation of $\widehat{A}_{\tau}(\nu)$ is given by
$
\rel\left\{\widehat{A}_{\tau}(\nu)\right\}={\cal R}_\tau(\nu)=o( {\cal K}_{\tau}(\nu;N,\Delta t)).
$
\end{enumerate}
\end{theorem}
\begin{proof}
See Appendix A in supplementary material.
\end{proof}
We additionally 
define
\begin{equation}
{\cal L}_\tau(\nu)=\frac{\left[{\cal N}_\tau^{(k)}(\nu)\right]^2}{{\cal K}_\tau(\nu)}=
\frac{1}{4}(N-|\tau|)^{1+4(\delta^{(k)}-\delta)}\left(\frac{1}{2 }-|\nu|\Delta t  \right)^{-1},\;{\cal D}(L)=\cup_k {\cal D}_k(L).
\end{equation}

We have now established that if $(\nu,\tau)$ is {\em not} sufficiently close to one of the points in the set $\left\{\left(\nu_k^{(0)},\tau_k^{(0)}\right) \right\}$ then the expectation gives a negligible contribution.
On the other hand if $(\nu,\tau)$ is sufficiently close to one of the points in the set $\left\{\left(\nu_k^{(0)},\tau_k^{(0)}\right) \right\}$ then 
$
\E\left\{\widehat{A}_{\tau}(\nu) \right\}
$
is asymptotically divergent as ${\cal L}_\tau(\nu)$ becomes unbounded with increasing $N$.
The variance of the EMAF in either of these two cases  is non-negligible.
Solely characterising the first and second moment of the AF is not sufficient to design an inference procedure unless we wish to use method of moments estimation, which is in general undesirable.
We now aim to write down the (approximate) distribution of the EMAF to be able to do inferences more efficiently. We introduce the normalized AF given by:
\begin{equation}
\label{whatAN}
\widehat{A}_{\tau}^{(N)}(\nu)=
\frac{\widehat{A}_{\tau}(\nu)}{{\cal K}_{\tau}^{1/2}(\nu;N,\Delta t)}.\end{equation}
This normalized version of the AF is sensible considering  $\var\{\widehat{A}_{\tau}^{(N)}(\nu)\}={\cal O}(1)$. If we are sufficiently far from ${\cal D}={\cal D}(L)$ (for a reasonably chosen $L\in{\mathbb{R}}^+$) then the expected modulus square of  $\widehat{A}_{\tau}^{(N)}(\nu)$
is ${\cal O}(1)$ while evaluated at the present components in ${\cal D}$ will be increasing with ${\cal N}$.

Referring to Eqn \eqref{eqn:EMAF} we can rewrite $\widehat{A}_{\tau}(\nu)$
as a sum of $\widehat{M}_\tau(t_n)$.
\cite{Hindberg2009} discussed the distribution of $\widehat{A}_{\tau}(\nu)$
for an underspread process. We have in this paper assumed that $\{X_n\}_n$ is a sample from a Gaussian process, and so $\widehat{M}_\tau(t_n)$ is a quadratic form and will be a weighted sum of two $\chi^2_1$'s \cite{Johnson}.
$\widehat{A}_{\tau}(\nu)$ is therefore an aggregated sum of 
correlated weighted sum of two $\chi^2_1$'s. To ensure that $\{\widehat{A}_{\tau}(\nu)\}$
is a collection of Gaussian random variables the correlation of $\{X_n\}$ across $n$ must be moderate. We also want to ensure that too much mass does not concentrate on one of the aggregated variables. 
Under such conditions $\widehat{A}_{\tau}(\nu)$ is approximately Gaussian. 
\begin{figure}[!htbp]
\begin{center}
\begin{minipage}[]{0.39\textwidth}
\centering
\includegraphics[width=\textwidth]{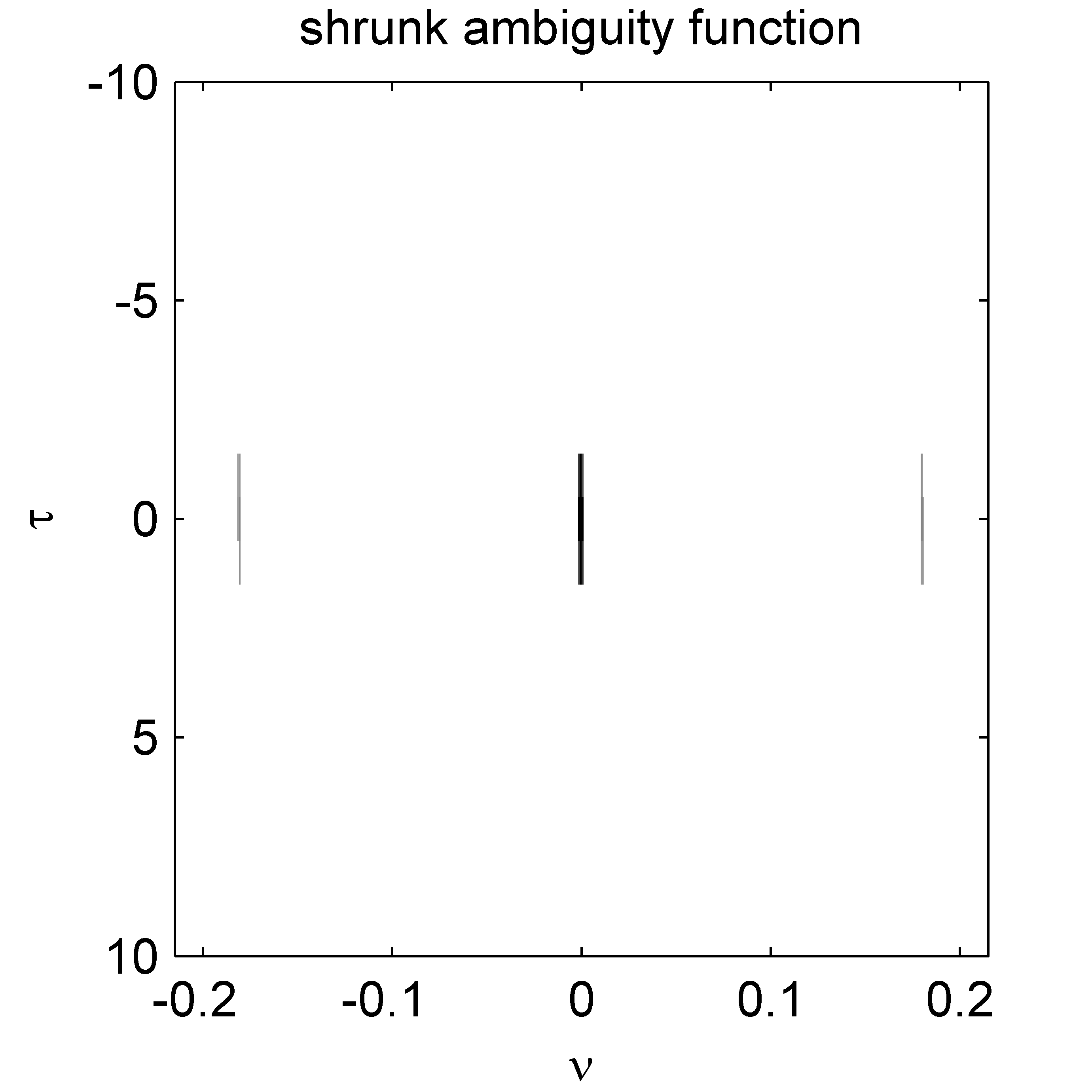}
(a)
\end{minipage}
\begin{minipage}[]{0.39\textwidth}
\centering
\includegraphics[width=\textwidth]{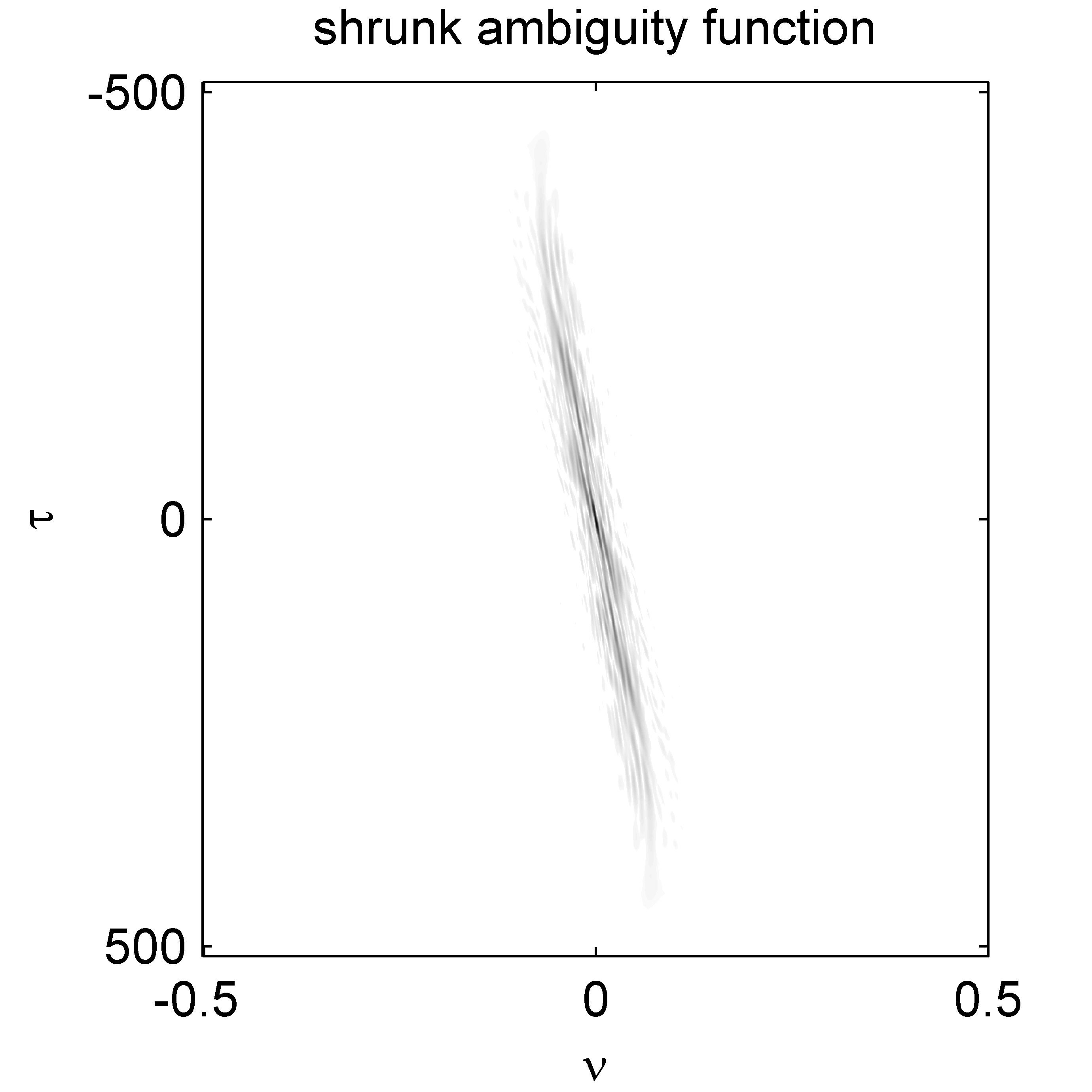}
(b)
\end{minipage}
\begin{minipage}[]{0.39\textwidth}
\centering
\includegraphics[width=\textwidth]{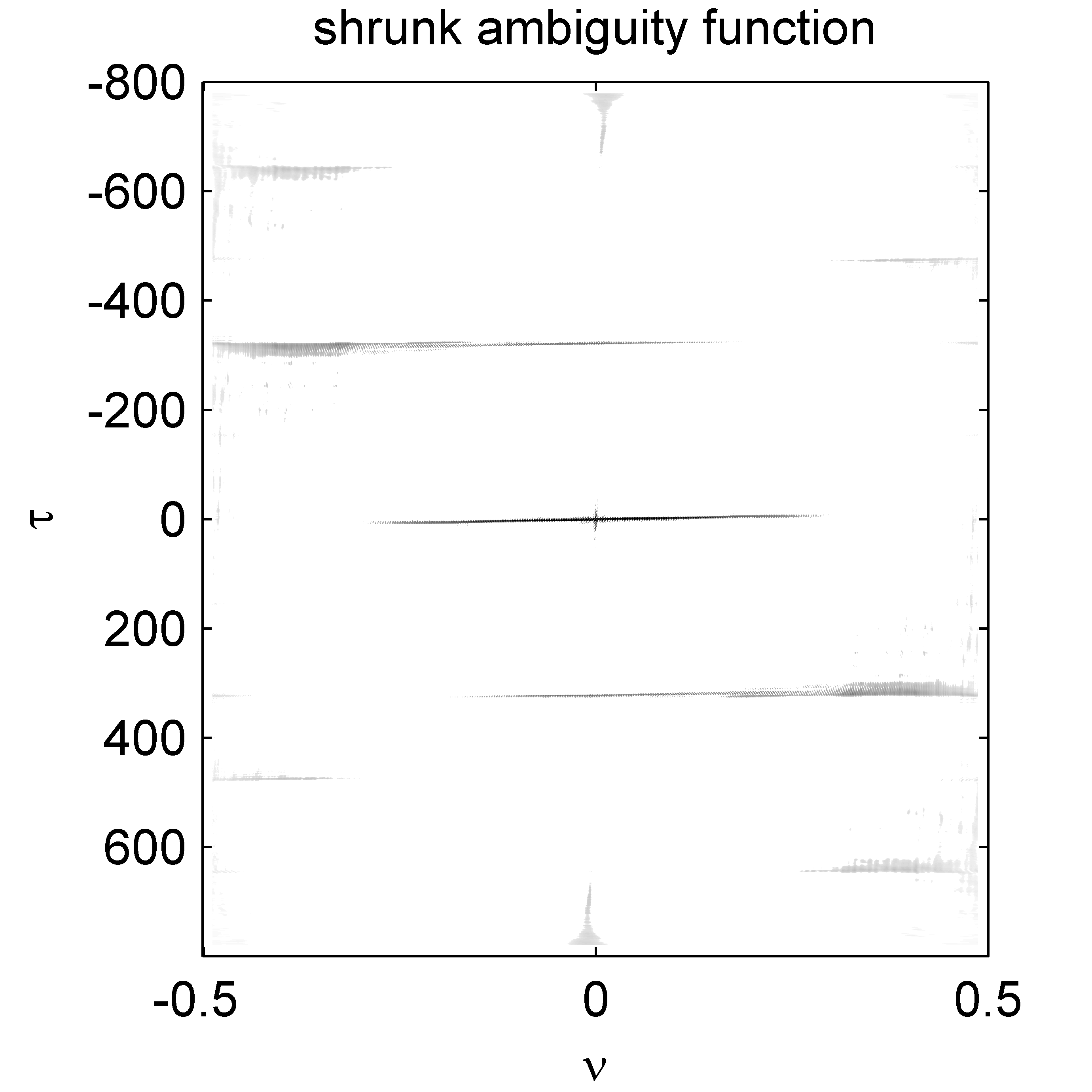}
(c)
\end{minipage}
\begin{minipage}[]{0.39\textwidth}
\centering
\includegraphics[width=\textwidth]{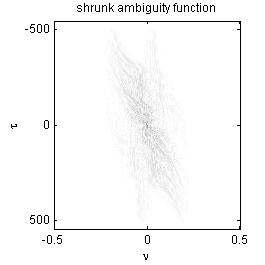}
(d)
\end{minipage}
\end{center}
\caption{Estimated ambiguity function of the realizations; the first is the aggregation of a locally stationary and a cyclostationary signal, the second a simulated blood flow signal, the third the pipistrellus signal and the fourth is the oceanographic signal.
\label{fig2}}
\end{figure}
\begin{corollary}{Magnitude of Ambiguity Surface\label{ConcAmbSurfII}}\\
The distribution of of the normalized AF represented in terms of its real and imaginary parts  or
\begin{equation}
\widehat{A}_{\tau}^{(N)}(\nu)=\widehat{C}_{\tau}^{(N)}(\nu)-i \widehat{D}_{\tau}^{(N)}(\nu)=
\widehat{Q}_\tau(\nu) e^{-i\vartheta_\tau(\nu)},
\end{equation} 
takes the form
\begin{enumerate}
\item If 
$(\nu,\tau)\notin {\cal D}$ then
$
\widehat{Q}_\tau^2(\nu)\sim \frac{{\cal V}_{\tau}(\nu)}{2}\chi^2_2.$
\item If $(\nu,\tau)\in {\cal D}$ then 
$\frac{\widehat{Q}_\tau^2(\nu)}{{\cal L}_{\tau}(\nu)}=\left|B_{\tau}(\nu)\right|^2+O\left(\frac{1}{{\cal L}_{\tau}(\nu)}\right),$
and
$
\var\left\{\widehat{A}_{\tau}^{(N)}(\nu)\right\}={\cal V}_{\tau}(\nu)$, 
$\rel\left\{\widehat{A}_{\tau}^{(N)}(\nu)\right\}={\cal R}_{\tau}(\nu).$ We have $\left|B_{\tau}(\nu)\right|^2=\left|A_\tau(\nu){\cal J}_1(\delta^{(k)};\{T(\tau)(\nu-\nu_0^{(k)}),\tau-\tau_0^{(k)}\}) \right|^2/({\cal L}_\tau(\nu){\cal K}_\tau(\nu))={\mathcal O}(1).$
\end{enumerate}
\end{corollary}
\begin{proof}
This follows by direct calculation from Theorem \ref{ConcAmbSurf}.
\end{proof}
We therefore expect to observe a mixture as the distribution of the EMAF, as is born out by practical examples, see Figure \ref{fig5}.
We let $\rho_{N,\Delta t}$ be the probability that a random coefficient is non-zero, with the desired sampling.
For a white noise process the expected AF
decays like $\tau^{-2}$ and $\nu^{-2}$ for sufficiently large $\nu$ and $\tau$.
If we add up the coefficients for which 
$\left|k\tau\right|\le N^{1/2}$
then these are
$2N+{\cal O}(\log(N))$.
The total number of coefficients are $4N^2$ and so the probability of hitting such a coefficient is ${\cal O}(1/N)$. 
For stationary and uniformly modulated processes the situation is very much similar. 

\begin{proposition}{Distribution of the Empirical Ambiguity Surface}\\
If we draw a coefficient at random at local frequency and time $(\nu,\tau)\neq (0,0)$ then if the random process satisfies ${\cal R}_{\tau}(\nu)\equiv 0$ we have that $\left|\widehat{A}_{\tau}^{(N)}(\nu)\right|^2$ is distributed as a mixture of central and non-central $\chi^2$'s or:
\begin{eqnarray}
\label{mixture1}
\left|\widehat{A}_{\tau}^{(N)}(\nu)\right|^2\sim (1-\rho_{N,\Delta t}) \frac{{\cal V}_{\tau}(\nu)}{2}\chi^2_2+
\rho_{N,\Delta t}\frac{{\cal V}_{\tau}(\nu)}{2}\chi^2_2\left({\cal L}_{\tau}(\nu)\frac{B_{\tau}^2(\nu)}{{\cal V}_{\tau}(\nu)}\right)
\end{eqnarray} 
\end{proposition}
\begin{proof}
This result follows by direct calculation from Theorem \ref{ConcAmbSurf} and Lemma \ref{ConcAmbSurfII}.
\end{proof}
Even if we expect decay in 
${\cal V}_{\tau}(\nu)$ for larger $\tau$ and $\nu$, this quantity will normally take a typical value of $\overline{\cal V}$ and so we may state the following result.
\begin{figure}[!htbp]
\begin{center}
\begin{minipage}[]{0.39\textwidth}
\centering
\includegraphics[width=\textwidth]{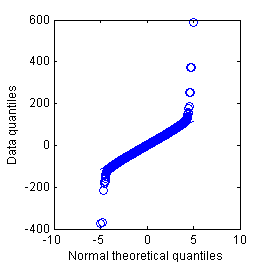}
(a)
\end{minipage}
\begin{minipage}[]{0.39\textwidth}
\centering
\includegraphics[width=\textwidth]{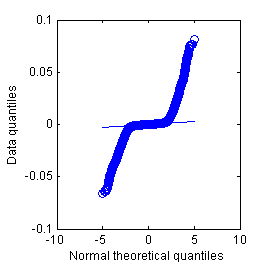}
(b)
\end{minipage}
\begin{minipage}[]{0.39\textwidth}
\centering
\includegraphics[width=\textwidth]{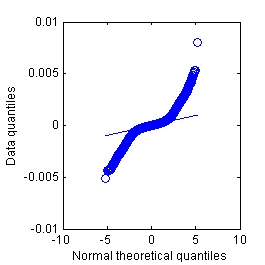}
(c)
\end{minipage}
\begin{minipage}[]{0.39\textwidth}
\centering
\includegraphics[width=\textwidth]{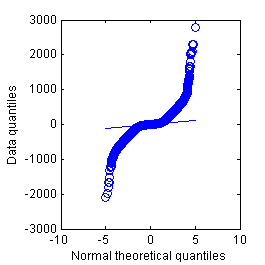}
(d)
\end{minipage}
\end{center}
\caption{QQ-plots of the real and imaginary parts of the renormalized ambiguity function, using $\widehat{\cal V}$ to estimate the superimposed lines. (a) is showing the simulated mixture of a cyclostationary and locally stationary process, (b) the simulated blood flow signal, (c) the chirping bat signal, and (d) the oceanography signal. 
\label{fig5}}
\end{figure}
\begin{corollary}{Distribution of Empirical AF}\label{distoffAmb}\\
If we draw a coefficient at random at local frequency and time $(\nu,\tau)\neq (0,0)$ then if the random process satisfies ${\cal R}_{\tau}(\nu)\equiv 0$ we have that $\left|\widehat{A}_{\tau}^{(N)}(\nu)\right|^2$ is distributed as a mixture or:
\begin{alignat}{1}
\label{mixture2}
\left|\widehat{A}_{\tau}^{(N)}(\nu)\right|^2\sim (1-\rho_{N,\Delta t}) \frac{\overline{\cal V}}{2}\chi^2_2+
\rho_{N,\Delta t}\frac{\overline{\cal V}}{2}\chi^2_2\left({\cal L}_{\tau}(\nu;N,\Delta t)\frac{B_{\tau}^2(\nu)}{\overline{\cal V}}\right).
\end{alignat}
\end{corollary}

This gives a distribution which is a simple mixture distribution of a central $\chi^2_2$ and a non-central $\chi^2_2$. It is still heavily overparameterised even if the variances are similar across $(\nu,\tau)$ because in addition to the two parameters $\rho_{N,\Delta t}$ and $\overline{\cal V}$, the sequence $\{B_{\tau}^2(\nu)\}$ {\em is not known}. If however $\rho_{N,\Delta t}$ is reasonably small, than the number of $B_{\tau}^2(\nu)$ that we have to learn is limited. We notice that this falls into the framework of the modelling adopted by \cite{Johnstone2004} of ``needles and haystacks''. Some coefficients
$\{\widehat{A}_{\tau}^{(N)}(\nu)\}$ have a non-zero mean but we do not know {\em which} coefficients these are, or how frequent they are. Our ability to learn the AF
will depend on how $\rho_{N,\Delta t}$ changes with $N$ and $\Delta t$. Realistically we assume that there {\em is} growing complexity in the time series and that this {\em is} quantified by  $\rho_{N,\Delta t}$. 
Unlike the case of nonparametric regression ({\em e.g.} \cite{Johnstone2004}) we do not have a set of $N$ uncorrelated coefficients and $N$ observations, but normally have to estimate $2N\times 2N-1$ correlated coefficients from $N$ observations. 
\begin{proposition}{Correlation Structure of $A_{\tau}(\nu)$\label{corrstruct}}\\
If $\{X_n\}$ is a white noise process then the covariance of the AF is
{\small \begin{eqnarray*}
\cov\left\{\widehat{A}_{\tau_1}(\nu_1),\widehat{A}_{\tau_2}(\nu_2)\right\}
&=&\left[\frac{(-1)^{j_1-j_2}e^{-2\pi \Delta t\max(\nu_1,\nu_2)) i (\tau_1-\tau_2)}-e^{2\pi \Delta t \max(\nu_1,\nu_2) i (\tau_1-\tau_2)}}{i2\pi  (\tau_1-\tau_2)}\right]\\ &&\times(N-\max(\tau_1,\tau_2))\sigma^4e^{i\pi (j_2-j_1)}\delta_{j_1,j_2}
\end{eqnarray*}}
\end{proposition}
\begin{proof}
See Appendix C.
\end{proof}
If we fix $\tau_1=\tau_2=\tau$ then if we take $\nu_k=\frac{k}{N-|\tau|}$ we retrieve uncorrelated coefficients. If $\tau$ is ${\cal O}(1)$, we could take $\nu_k=\frac{k}{N}$. For white noise we approximately obtain a set of $N$ uncorrelated random variables by taking $\{
\widehat{A}_{\tau}(\nu_k)\}_k$. This argument can be repeated for any order one value of $\tau$ we should choose. The full set of coefficients $\{
\widehat{A}_{\tau}(\nu_k)\}_{k,\tau}$ is almost like many redundant collections of uncorrelated random variables. It is reasonable with the model of an ambiguity sparse process that the full set of coefficients $\{
\widehat{A}_{\tau}(\nu_k)\}_{k,\tau}$ are approximately like ${\cal O}(N)$ collections of uncorrelated coefficients; otherwise the procedure is like a composite likelihood method. Other choices ({\em e.g.} $\nu_k=\frac{k}{2N}$) will also produce such redundant collections.

\section{Inference Methods}\label{infmethod}
The determination of the statistical properties of the 
EMAF in Section \ref{methmoments} has now put us in the fortunate position where we can propose inference methods. For ease we now {\em model} each individual $B_{\tau}(\nu)$. 
To remove explicit dependence on ${\cal L}$ we now define
$
{\cal B}_{\tau}(\nu)=\sqrt{{\cal L}_{\tau}(\nu)}B_{\tau}(\nu).$
 We model the normalized AF to be
\begin{eqnarray}
{\cal B}_{\tau}(\nu)&\sim& N^{C}\left(0,\sigma^2_{\cal L}\right)=
\cQ_{\tau}(\nu)e^{-i\phi_\tau(\nu)},
\end{eqnarray}
and note that $\sigma^2_{\cal L}$ increases directly with $\cN_{\tau}(\nu)$. $N^{C}\left(\cdot,\cdot\right)$
refers to the complex-proper Gaussian distribution \cite{Schreier}.
We may wish to put our belief regarding time-frequency structure into this prior, but notice that modelling {\em local} time and frequency $(\nu,\tau)$ {\em directly} controls the smoothness of the global time and frequency, rather than modelling global time and frequency directly, see for example work by \cite{Godsill}, which constrains the time-varying spectrum to sparsity. Furthermore for some (very few) locations the above prior is not reasonable: {\em e.g.} $A_0(0)$ will always be real-valued, but the effects of modelling this single coefficient incorrectly are negligible as it will anyway always be retained.  
We define the likelihood in terms of the parameters
\begin{equation}
\bpsi=\begin{pmatrix} \overline{\cal V} &\rho & \sigma^2_{\cN}
\end{pmatrix}^T,\quad \bm{\mathcal{B}}=\{{\mathcal{B}}_{\tau}(\nu_k) \}_{(\tau,k)}
,\quad \bm{\mathcal{Q}}=\{{\mathcal{Q}}_{\tau}(\nu_k) \}
,\quad
{\mathbf{A}}=\{\widehat{A}^{(N)}_{\tau}(\nu_k)\}.
\end{equation}
These vectors are defined over indices $\tau \in [-(N-1),(N-1)]$ and $\nu_k=\frac{k}{2N\Delta t}$ for $k=-N,\dots,N$. 
We write
$
\widehat{A}^{(N)}_{\tau}(\nu_k)=\widehat{Q}_{\tau}(\nu_k)e^{-i\vartheta_\tau(\nu)},\quad \mathbf{Q}=\{\widehat{Q}_{\tau}(\nu_k) \},$\footnote{To calculate $\widehat{A}^{(N)}$ we need to know ${\mathcal K}_\tau(\nu)$ and therefore $\delta^{(k)}$. For most processes $\delta^{(k)}=\frac{1}{2}$ is a reasonable choice.}
and define the ``ambiguity likelihood''.
\begin{definition}{Ambiguity Likelihood}\\
We define the ambiguity likelihood for the ambiguity surface $\widehat{A}_{\tau}(\nu)$ to be
{\tiny
\begin{alignat}{1}
\label{eq:amblik}
L(\bpsi,\bm{\mathcal{B}},{\mathbf{A}})
&=\prod_{\tau,k}\left\{
\frac{(1-\rho)}{\pi\overline{\cal V}} e^{-\frac{\left|\widehat{A}^{(N)}_{\tau}(\nu_k)\right|^2}{\overline{\cal V}}}\delta\left({\cal B}_{\tau}(\nu_k)\right)+
\frac{\rho}{\pi^2\overline{\cal V}\sigma_\cN^2} e^{-\frac{\left|\widehat{A}^{(N)}_{\tau}(\nu_k)-{\cal B}_{\tau}(\nu_k)\right|^2}{\overline{\cal V}}}e^{-\frac{\left|{\cal B}_{\tau}(\nu_k)\right|^2}
{\sigma^2_\cN}}\right\}
\end{alignat}}
\end{definition}
We have coupled the sparsity of the ambiguity surface between the real and imaginary components. 
Secondly \eqref{eq:amblik} is like a true likelihood for a subset of $N$ coefficients if we have chosen $(\nu_k,\tau)\in \Upsilon_1$ where $\Upsilon_1$ is chosen to break up the correlation, averaged over the choices of coefficients we could have taken, i.e. the disjoint sets $\{\Upsilon_l\}$ such that $\cup_l \Upsilon_l$ gives the full set of $\tau$ and $\nu_k$ that we have calculated, see Proposition \ref{corrstruct}. 
\begin{definition}{Ambiguity Marginal Likelihood}\\
We define the ambiguity marginal likelihood to be:
\begin{alignat}{1}
\label{eq:amblikmarg}
L(\bpsi;{\mathbf{A}})&=\prod_{k,\tau}
\left\{\frac{2(1-\rho)}{\overline{\cal V}}\hat{Q}_{\tau}(\nu_k) e^{-\frac{\hat{Q}^{2}_{\tau}(\nu_k)}{\overline{\cal V}}}+\frac{2\rho}{\overline{\cal V}+\sigma^2}\hat{Q}_{\tau}(\nu_k)
 e^{-\frac{\hat{Q}^{2}_{\tau}(\nu_k)}{\overline{\cal V}+\sigma^2}}\right\}.
\end{alignat}
\end{definition}
This form is derived in Appendix D by integrating over the other variables.
We define $
\widehat{\bpsi}=\arg_{\bpsi}\max \left\{L(\bpsi;{\mathbf{A}})\right\}.
$
This maximum can be found by numerical optimization methods.
We can now find posterior estimators of ${\cal B}_\tau(\nu),$ following
\cite{Johnstone2004,Wang}, and using the posterior median estimator (this has the advantage of corresponding to hard thresholding for certain ranges of the parameters). 
We wish to calculate the posterior distribution of the ambiguity coefficient, given the observed ambiguity coefficient.
We start from \eqref{eq:amblik} for a single coefficient and get:
{\small \begin{eqnarray}
\label{joint}
f\left({\cal B}_{\tau}(\nu_k),\widehat{A}_{\tau}(\nu_k)\right)&=&
\frac{(1-\rho)}{\pi\overline{\cal V}} e^{-\frac{\left|\widehat{A}^{(N)}_{\tau}(\nu_k)\right|^2}{\overline{\cal V}}}\delta\left({\cal B}_{\tau}(\nu_k)\right)+
\frac{\rho}{\pi^2\overline{\cal V}\sigma_\cN^2} e^{-\frac{\left|\widehat{A}^{(N)}_{\tau}(\nu_k)-{\cal B}_{\tau}(\nu_k)\right|^2}{\overline{\cal V}}}e^{-\frac{\left|{\cal B}_{\tau}(\nu_k)\right|^2}
{\sigma^2_\cN}}.
\end{eqnarray}}
Calculating
${\cal B}_{\tau}(\nu)|\widehat{A}_{\tau}(\nu)$ does not make sense, as we are then thinking of ${\cal B}_{\tau}(\nu)$ with the phase information of $\widehat{A}_{\tau}(\nu),$
and we shall instead estimate ${\cal Q}_{\tau}(\nu)$ conditional on observing $\widehat{Q}_{\tau}(\nu)$, subsequently shrinking $\hat{Q}_{\tau}(\nu)$ irrespective of the phase distribution of $\widehat{A}_\tau(\nu)$.
The posterior distribution of the ambiguity coefficient at a given location is given by:
\begin{eqnarray}
\nonumber
f({\mathcal Q}_{\tau}(\nu)|\hat{Q}_{\tau}(\nu))&=&\left(1-\rho_{\tau}^{({\mathrm{post}})}(\nu)
\right)\delta\left({\mathcal{Q}}_{\tau}(\nu)\right)+\rho_{\tau}^{({\mathrm{post}})}(\nu)
f_1\left({\mathcal Q}_{\tau}(\nu)|\hat{Q}_{\tau}(\nu) \right)
\\
\nonumber
f_1\left({\mathcal Q}|\hat{Q}\right)&=&\frac{2}{\overline{\cal V}\lambda}
e^{-\lambda\frac{\hat{Q}^2}{ \overline{\cal V}}}
e^{-\frac{{\cal Q}^2}{\lambda\overline{\cal V}}}{\cal Q}
J_0\left(-i\frac{2{\cal Q}\hat{Q}}{\overline{\cal V}}\right)\approx N\left(\lambda \hat{Q},\frac{1}{2}\lambda\overline{\cal V}\right),\quad
\lambda=\frac{\sigma^2}{\sigma^2+\overline{\cal V}}.
\end{eqnarray}

\begin{proof}
See Appendix D, 
and the posterior probability is given by
\[\rho_{\tau}^{({\mathrm{post}})}(\nu)=\rho \frac{e^{-\frac{\hat{Q}^2_\tau(\nu)}{\sigma^2+\overline{\cal V}}}}{\sigma^2+\overline{\cal V}}\left(\rho \frac{e^{-\frac{\hat{Q}^2_\tau(\nu)}{\sigma^2+\overline{\cal V}}}}{\sigma^2+\overline{\cal V}}+(1-\rho) \frac{e^{-\frac{\hat{Q}^2_\tau(\nu)}{\overline{\cal V}}}}{\overline{\cal V}}\right)^{-1}.\]
\end{proof}
With this distribution a convenient estimator is the posterior median, see \cite{Johnstone2004}, and we take
$
\eta_{\tau}(\nu_k)=\Phi\left(-\sqrt{2\widehat{\lambda}}\frac{\hat{Q}}{\widehat{\cal{V}}^{1/2}}\right).$
The posterior median estimator solves 
\begin{eqnarray}
F(\cQ^{(\mathrm{median})}|\hat{Q})&=& \frac{1}{2}
\label{blahis}
\approx (1-\rho^{(\mathrm{post})}_{\tau}(\nu_k;\hat{Q}))+\rho^{(\mathrm{post})}_{\tau}(\nu_k)
\Phi\left(\frac{|{\cal Q}|^{(\mathrm{median})}-\lambda \hat{\cal Q}}{\sqrt{\frac{\lambda}{2}}{\overline{\cal V}^{1/2}}}\right),
\end{eqnarray}
and
$\zeta_{\tau}(\nu_k)=\rho^{(\mathrm{post})}_{\tau}(\nu_k)\eta_{\tau}(\nu_k).$
If $\zeta_{\tau}(\nu_k)\le \frac{1}{2}$ then the posterior odds are in favour of a zero-valued coefficient and the coefficient's magnitude is estimated as zero, otherwise the median is found from \eqref{blahis}.
Thus with
\begin{eqnarray}
\Theta_{\tau}(\nu_k)\approx\left\{\begin{array}{lcr}
0 & {\mathrm{if}} & \zeta_{\tau}(\nu_k)\le\frac{1}{2}\\
\lambda+\Phi^{-1}\left(1-\frac{1}{2\rho^{(\mathrm{post})}_{\tau}(\nu_k)}\right)
\frac{\sqrt{\lambda}\widehat{\cal V}}{\sqrt{2}\left|A^{(N)}_{\tau}(\nu_k) \right|}
& {\mathrm{if}} & \zeta_{\tau}(\nu_k)>\frac{1}{2}
\end{array} \right.
\label{rule}
\end{eqnarray}
the posterior median estimator is therefore
\begin{equation}
\label{postmed}
\widehat{A}^{({\mathrm{eb}})}_{\tau}(\nu_k)=\Theta_{\tau}(\nu_k)\widehat{A}_{\tau}(\nu_k).
\end{equation}
The estimator in \eqref{postmed} is a shrinkage estimator, as long as $\left|\Theta_{\tau}(\nu_k)\right|\le 1$ (which will be the case see {\em e.g.} \cite{Johnstone2004}), which converts to a smoothing operation in the dual-time domain:
\begin{eqnarray}
\label{ebby}
\widehat{M}^{({\mathrm{eb}})}_{\tau}(t_n)=\Delta t\sum_{m=0}^{N-1} \theta_{\tau}(s_m-t_n)\widehat{M}_{\tau}(s_m).
\end{eqnarray}
Eqn. \eqref{ebby} directly mirrors \eqref{kernel:est}, except $\{\omega_\tau(t_n)\}$ has been replaced by the {\em data dependent shrinkage} procedure  $\{\theta_\tau(t_n)\}$. Thus we can interpret \eqref{postmed} as a data-driven smoothing of the raw sample moments
$\{\widehat{M}_\tau(t_n)\}$, using the estimated sequence $\{\theta_\tau(t_k) \}$. To investigate the smoothing function more clearly we note that
$
\theta_\tau(t_n)=\int_{-\frac{1}{2\Delta t}}^{\frac{1}{2\Delta t}}\Theta_\tau(\nu)e^{i2\pi \nu t}\,d\nu,
$
and we see that for each fixed value of $\tau$ we define a different smoothing kernel. 
\begin{figure}[!htbp]
\begin{center}
\begin{minipage}[]{0.39\textwidth}
\centering
\includegraphics[width=\textwidth]{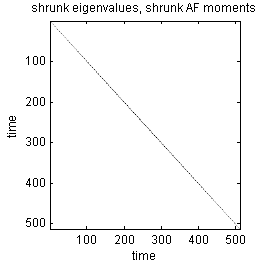}
(a)
\end{minipage}
\begin{minipage}[]{0.39\textwidth}%
\centering
\includegraphics[width=\textwidth]{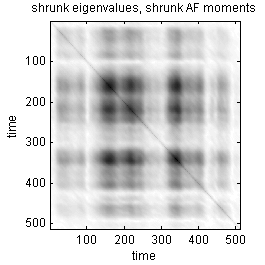}
(b)
\end{minipage}
\begin{minipage}[]{0.39\textwidth}
\centering
\includegraphics[width=\textwidth]{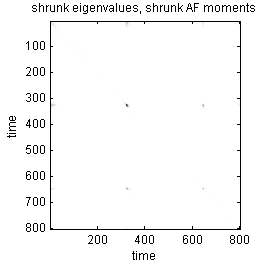}
(c)
\end{minipage}
\begin{minipage}[]{0.39\textwidth}
\centering
\includegraphics[width=\textwidth]{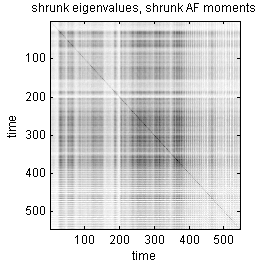}
(d)
\end{minipage}
\end{center}
\caption{Estimated covariance structure of the signals. These are estimated from the shrunk ambiguity function, and corrected to non-negative eigenvalues. Subplot (a) is showing the simulated mixture of a cyclostationary and locally stationary process, (b) the simulated blood flow signal, (c) the chirping bat signal, and (d) the oceanography signal.
\label{fig3}}
\end{figure}
As $\rho$ decreases the probability of thresholding a larger portion increases, and so most of the ambiguity domain is zeroed out. If $\rho$ is anticipated to be a decreasing function in $N$ the procedure will be consistent.
See Figure \ref{fig9} for an example of four different kernels we retrieve for four different examples. The first subplot shows an example of an aggregation of a cyclostationary and a locally stationary plot. The estimators limit the support in $\tau$, and shows seasonality in $t$ as the cyclostationary features are reinforced. Subplot (b) shows a chirping signal, where it is advantageous to use all lags at the same global time, and similar features are replicated in the meddy signal (subplot (d)). For the chirping acoustic signal there is clear selectivity in both global and local time.

Furthermore note that $\rho$ is the probability that we will come up with a non-zero contribution. It is not equal to the area of ${\cal D}$ divided by ${\cal L}_{\tau}(\nu)$ because we are not sure that all of ${\cal D}$ is actually supported. 
It may also become necessary to relax \eqref{mixture2} to allow for different distributions of variances.  We can instead (if necessary) take a mixture model with $P$ components and taking $\{\overline{\cal V}_p\}_{p=1}^P,$
allowing for stronger and weaker signals.

\section{Estimating the Time-Varying Spectrum \& Valid Second Order Forms}
For stationary time series the spectrum is an inherently important analysis tool, see {\em e.g.} \cite{PercivalWalden1993}. It shows the distribution of energy of a time series across frequencies, thus characterising the time series, permits inference via the Whittle likelihood, allows us to check if a posited autocovariance is valid, and we may even forecast future values directly from the frequency domain, see {\em e.g.} \cite{Haywood1997}. Eqn 
\eqref{moments1} defined a time-varying spectrum by Fourier transforming the local moment function. Whilst this appears to be a self-evidently simple extension of the usual spectrum, it fails to satisfy a number of desiderata, see {\em e.g.} the full discussion in \cite{loynes} or \cite{cohen}, such as positivity. It can even be proved that any bilinear representation of the spectral content of the signal {\em must} fail some of the desiderata that are required for a time-varying spectrum. We start by defining an estimator of the 
time-varying moments of $Z_t$ by
\begin{eqnarray}
\widehat{M}_\tau^{({\mathrm{ebay}})}(t_n)&=&\int_{-\frac{1}{2\Delta t}}^{\frac{1}{2\Delta t}}\widehat{A}_\tau^{({\mathrm{ebay}})}(\nu)e^{2i\pi \nu t_n}\,d\nu
\approx \frac{1}{2N\Delta t}\sum_{k=-N}^{N-1}
\widehat{A}_\tau^{({\mathrm{ebay}})}(\nu_k)e^{2i\pi \nu_k t_n},\quad  \nu_k=\frac{k}{2N\Delta t}.
\end{eqnarray}

In section \ref{fourfaces} we discussed various definitions of the time-varying spectrum corresponding to special cases of\footnote{The indicator function can be removed and $M_\tau(t_n^{\alpha})$ interpolated instead.}
\begin{equation}
\label{bilinear}
S^{\alpha,\omega}(t_n,f)=\Delta t^2 \sum_{\tau=-(N-1)}^{N-1}\sum_{k=0}^{N-1} \omega_{\tau}\left(t_k-t_n\right) M_\tau (t_n^{\alpha})e^{-2i\pi \tau f}
I\left(t_n^\alpha \in {\mathbb{Z}}\right).
\end{equation}
By  changing the definition of $\{\omega_{\tau}\left(t\right)\}$ and $\alpha$ we will obtain different Fourier representations of the sequence $M_\tau (t_n^{\alpha})$. The utility of any particular bilinear representation will depend on the analysis problem in question. There is at times in signal processing confusion as to the relation of Eqn \eqref{smooth} (or Eqns \eqref{kernel:est} \& \eqref{ebby}, chosen to smooth $M_\tau(t_n)$) with \eqref{bilinear} (or Eqn \eqref{SFT_rep}, defining a different time-frequency representation). The reason for this is simple; the equations appear to be performing the identical action. Because any bilinear representation ({\em e.g.} \eqref{bilinear}) $S^{\alpha_1,\omega_1}(t_n,f)$ can be written as a convolution of any other $S^{\alpha_2,\omega_2}(t_n,f)$ (see {\em e.g.} \cite[p.~187]{Flandrin99}), at times any bilinear representation is viewed as an estimator of any other, see {\em e.g.} \cite[p.~299]{Scharf2005}. However in signal processing the sequence $\{\omega_{\tau}\left(t_k\right)\}$
is  chosen to improve the visual appearance of $S^{\alpha,\omega}(t_n,f)$ rather than considering the estimation properties of a sample version -- unlike $\varpi_\tau(t_k)$ and $\theta_\tau(t_k).$ \cite{Sayeed95} is a notable exception, but no practical estimation schemes are proposed in that paper, as knowledge of the higher order moments are required for implementing these ideas. We have {\em separated} the estimation of $M_\tau (t_n^{\alpha})$ from the (mathematical) choice of representation of this object, once estimated.

What we have failed to discuss is the validity of any given estimated covariance sequence $\{\widehat{M}_\tau (t_n^{\alpha})\}$. If we arrange $M_\tau (t_n^{\alpha})$ to represent the covariance of the vector
${\mathbf{Z}}=\begin{pmatrix} Z_{t_1} & \dots Z_{t_N}
\end{pmatrix},$
it follows that
$
\bm{\Lambda}=\E\left\{{\mathbf{Z}}{\mathbf{Z}}^H\right\}=\bm{\Lambda}
\left(\{M_\tau (t_n^{\alpha})\}\right),
$
should be a valid covariance matrix, {\em e.g.} all its eigenvalues must be non-negative. For a stationary time series this can also be ensured by
requiring the Fourier transform of the autocovariance sequence to be non-negative, and it is commonly considered a desideratum also for the time-varying spectrum. For a nonstationary series it is unrealistic to assume that a time-varying spectrum will be non-negative. 

The raw method of moments estimator we started with, {\em e.g.}
$\widehat{\bm{\Lambda}}={\mathbf{Z}}{\mathbf{Z}}^H=\bm{\Lambda}
\left(\{\widehat{M}_\tau (t_n^{\alpha})\}\right),$
has only one eigenvalue corresponding to the total energy ${\mathbf{Z}}^H{\mathbf{Z}}$, and all other eigenvalues are identically zero. We create an alternative estimator of $\bm{\Lambda}$ from $\widehat{M}_\tau^{({\mathrm{ebay}})}(t_n)$, this yielding $\widehat{\bm{\Lambda}}^{({\mathrm{ebay}})}$.
Unfortunately $\widehat{\bm{\Lambda}}^{({\mathrm{ebay}})}$ does {\em not} necessarily have positive eigenvalues, nor is it in general sparse, where the latter could be used to ensure positivity. 
\begin{figure}[!htbp]
\begin{center}
\begin{minipage}[]{0.39\textwidth}
\centering
\includegraphics[width=\textwidth]{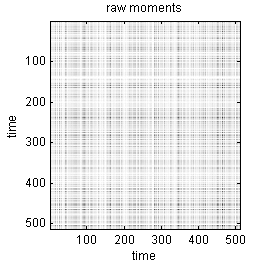}
(a)
\end{minipage}
\begin{minipage}[]{0.39\textwidth}
\centering
\includegraphics[width=\textwidth]{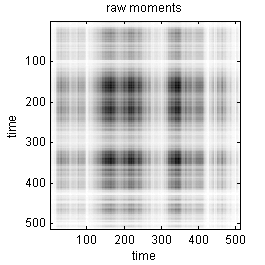}
(b)
\end{minipage}
\begin{minipage}[]{0.39\textwidth}
\centering
\includegraphics[width=\textwidth]{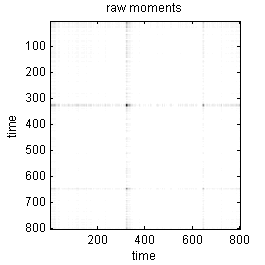}
(c)
\end{minipage}
\begin{minipage}[]{0.39\textwidth}
\centering
\includegraphics[width=\textwidth]{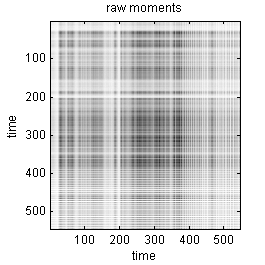}
(d)
\end{minipage}
\end{center}
\caption{Raw moments estimates of the covariance structure of the signals. Subplot (a) is showing the simulated mixture of a cyclostationary and locally stationary process, (b) the simulated blood flow signal, (c) the chirping bat signal, and (d) the oceanography signal. 
\label{fig4}}
\end{figure}
We may ``correct'' the estimator by two possible methods, as follows. First we calculate the eigendecomposition
$
\widehat{\bm{\Lambda}}^{({\mathrm{ebay}})}=\mathbf{U} \bm{\Upsilon}\mathbf{U}^H,$
and then we correct it using one of these two methods:
$\widehat{\bm{\Lambda}}^{({\mathrm{ebay}})}_1=\mathbf{U} \left[\bm{\Upsilon}-\left\{\min_j {\upsilon}_j \right\}\mathbf{I}\right]\mathbf{U}^H$ or
$
\widehat{\bm{\Lambda}}^{({\mathrm{ebay}})}_2=\mathbf{U} \left[\bm{\Upsilon}\right]_+\mathbf{U}^H,
$
where $\{{\upsilon}_j\}$ are the eigenvalues of $\widehat{\bm{\Lambda}}^{({\mathrm{ebay}})}$,
and $\left[\bm{\Upsilon}\right]_+$ is thresholding the entries of the diagonal matrix at 0. Both estimators ({\em e.g.} $\widehat{\bm{\Lambda}}^{({\mathrm{ebay}})}_1$
and $\widehat{\bm{\Lambda}}^{({\mathrm{ebay}})}_2$) are valid covariance matrices whilst $\widehat{\bm{\Lambda}}^{({\mathrm{ebay}})}$ is not. The estimate of the time-frequency spectrum produced by either of these matrices is however very similar in most cases. Using 
$\widehat{\bm{\Lambda}}^{({\mathrm{ebay}})}$
without any adjustment is like estimating a variance to be negative, and is therefore not to be recommended. 

\section{Examples}
We consider both simulated and real data examples. The main purpose of this section is to show the performance of our proposed method, and using any particular choice of representation, as well as a given value of $\alpha$. We stress again that in our opinion there is no optimal choice of representation or $\alpha$, but rather each representation has clear advantages and disadvantages for different processes, once the local moments of the time series have been estimated.

Let us start with an extremely sparse signal, defined by
$X_t=Y_t^{(1)}+Y_t^{(2)}$, $Y_t^{(k)}=\sigma_t^{(k)}\sum_{l=0}^5 w_l^{(k)} \varepsilon_{t-l}$, $t=0,\dots,N-1$ for $k=1,2,$
with
$w_t^{(1)}=\begin{pmatrix} 1 &  0.33 & 0.266 & 0.2 & 0.133 & 0.066\end{pmatrix}^T,$\\
$w_t^{(2)}=\begin{pmatrix}1 & 0.5 & 0 & 0.3 & 0 & 0.1\end{pmatrix}^T$
and
$\sigma_t^{(1)}=\frac{1}{4}+(\frac{t}{512})(1-\frac{t}{512}),$
$\sigma_t^{(2)}=4\left|\sin(2\pi\times 0.09t )\right|.$
The first of these two processes is a locally stationary process and the second a cyclostationary process. Their aggregation is {\em neither} locally stationary {\em nor} cyclostationary at the sampling we are looking at the signal. We simulate the signal to be length $512$.
This signal was considered using universal thresholding in the ambiguity domain by  \cite{Hindberg2009}, which need not produce a valid estimator. When analyzing it with the mixture model the estimated $\rho$ is very small; $\widehat{\rho}=4.6\times 10^{-5}$, a number  corresponding to about 50 pixels in the redundant representation. A plot of the estimated AF is given in Fig. \ref{fig2}(a), and the estimated moments are plotted in \ref{fig3}(a). The AF is in this instance very sparse indeed with only some contributions near the origin and the cyclostationary frequencies $0.18=2\times 0.09$. The difference between the raw and estimated moments ({\em e.g.} \ref{fig4}(a) vs \ref{fig3}(a)) is marked. We are here using one of the valid estimators of the covariance sequence, namely 
$\widehat{\bm{\Lambda}}_2^{({\mathrm{ebay}})}.$ Looking at raw characteristics of the renormalized real and imaginary part of the AF clearly fails to indicate mixture components, see Figure \ref{fig5}(a). This is because of the very high degree of sparsity of the AF. The sample Wigner distribution is much too noisy to be useful, see Figure \ref{fig6}(a). Smoothing using the method outlined in this paper keeps the high-frequency cyclostationary component, see Figure \ref{fig6}(b). 
The spectrogram fails to represent the cyclostationary component, see Figure \ref{fig6}(c). No matter what sophisticated inference procedure we would use on the raw spectrogram this procedure would clearly fail to retrieve the cyclostationary component. Because this is a simulated signal we can look at the normalized error of the proposed estimator, versus the normalized error of the raw covariance estimator -- see Figure \ref{fig8}(a)--(b). Comparing the estimated eigenvalues with the truth (Figure \ref{fig8}(c)) shows that the ambiguity domain shrinkage is regularizing the eigenvalues substantially, and the negative of the eigenvalues have here been set to zero to make the estimator valid. Figure \ref{fig8}(d) shows how extremely noisy the raw moments are, compared with the estimated moments (down at the very bottom of the plot). As a final point of interest we show $\{\widehat{\theta}_\tau(t)\}$ in Figure \ref{fig9}(a). This demonstrates how the estimated ``smoothing filter'' is reinforcing cyclical patterns with the right period, but shrinking most possible correlations.
\begin{figure}[!htbp]
\begin{center}
\begin{minipage}[]{0.36\textwidth}
\centering
\includegraphics[width=\textwidth]{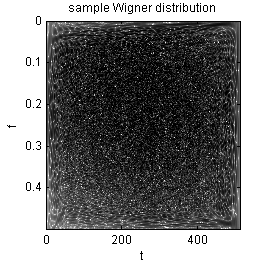}
(a)
\end{minipage}
\begin{minipage}[]{0.36\textwidth}
\centering
\includegraphics[width=\textwidth]{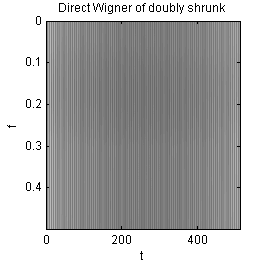}
(b)
\end{minipage}\\
\begin{minipage}[]{0.36\textwidth}
\centering
\includegraphics[width=\textwidth]{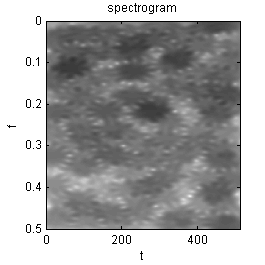}
(c)
\end{minipage}
\begin{minipage}[]{0.36\textwidth}
\centering
\includegraphics[width=\textwidth]{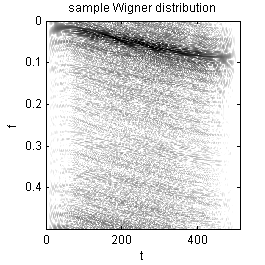}
(d)
\end{minipage}
\begin{minipage}[]{0.36\textwidth}
\centering
\includegraphics[width=\textwidth]{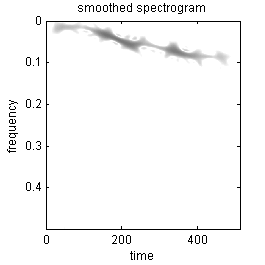}
(e)
\end{minipage}
\begin{minipage}[]{0.36\textwidth}
\centering
\includegraphics[width=\textwidth]{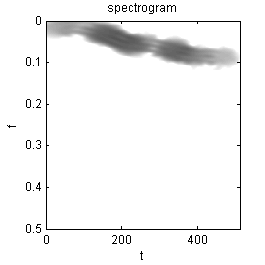}
(f)
\end{minipage}
\end{center}
\caption{Time-frequency representations of some of the sample signals. Subplot (a) is the raw Wigner distribution on a dB scale of the simulated mixture of a cyclostationary and locally stationary process, (b) is the smoothed Wigner distribution of the moments estimated from Ambiguity domain and (c) is the raw spectrogram from the data. Subplot (d) is the raw Wigner distribution on a dB scale of the simulated blood flow signal, (e) is the spectrogram of the moments estimated from Ambiguity domain and (f) is the raw spectrogram from the data. 
\label{fig6}}
\end{figure}
\begin{figure}[!htbp]
\begin{center}
\begin{minipage}[]{0.36\textwidth}
\centering
\includegraphics[width=\textwidth]{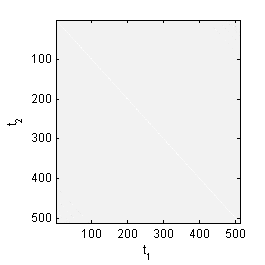}
(a)
\end{minipage}
\begin{minipage}[]{0.36\textwidth}
\centering
\includegraphics[width=\textwidth]{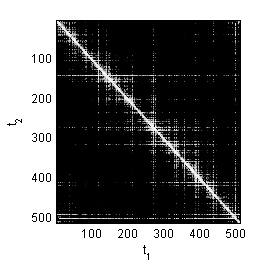}
(b)
\end{minipage}\\
\begin{minipage}[]{0.36\textwidth}
\centering
\includegraphics[width=\textwidth]{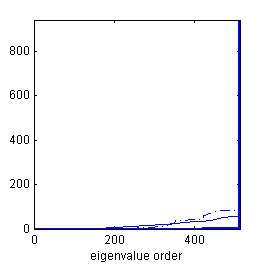}
(c)
\end{minipage}
\begin{minipage}[]{0.36\textwidth}
\centering
\includegraphics[width=\textwidth]{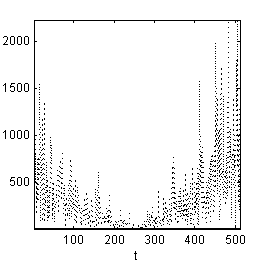}
(d)
\end{minipage}
\end{center}
\caption{Comparison of estimation error for a single realization of the aggregation of the cyclostationary and locally stationary process. The first shows the normalized estimation error using the EBAYES method combined with correcting the eigenvalues, the second the normalized estimation error using the raw moments. Subplot (c) shows the eigenvalues of these two methods. The solid thick line of one nonzero eigenvalue is the sample moments, the solid thin line the theoretical eigenvalues and the dot dashed is using the EBAYES method combined with correcting the eigenvalues. Subplot (d) shows the normalized estimation error for a single row of the matrix, where the EBAYES method is seen as superimposed over the $x$-axis, and the dotted line is the raw moments.
\label{fig8}}
\end{figure}

The second simulated example is a modified version of the simulated signal analysed in \cite{Olhede2003c}. The signal has simply been reduced in frequency range and subsampled, but takes the form:
\begin{equation}
Z_n=\Delta t \sum_{l=-M}^M\nu_{n-k} h_{k,n}+\eta_n,
\end{equation}
where $\{\nu_n\}$ and $\{\eta_n\}$ are independent white noise processes, $M$ is a positive integer, with $\{h_{k,n}\}_k$ a set of sequences defined for each value of $n$.
This signal has a strong oscillatory pattern with chirping period, due to the blood flow being basically ``forward''
during systole and ``reverse'' during diastole. In this action a number of frequencies are temporarily visited. This oscillatory structure is reinforced by \ref{fig2}(b) showing a sloping linear structure. Because of the oscillatory structure of the signal it shows presence in most of the time-time domain -- {\em cf} \ref{fig3}(b), but these are smoother versions of \ref{fig4}(b), again guaranteed to be positive semi-definite sequences by correcting eigenvalues. There is a clear two-population mixture in the renormalized data, {\em cf} Figure \ref{fig5}(b), and the estimated probability of belonging to the signal component in the mixture is 0.086. Looking at the raw Wigner distribution (Fig. \ref{fig6}(d)),
the spectrogram (Fig. \ref{fig6}(f)) and the smoothed representation of the estimated moments using the shrunken AF (Fig. \ref{fig6}(e), again using \eqref{bilinear} for nicer visual characteristics with $\omega_\tau(t)$
corresponding to an appropriate combination of Hermite windows, see \cite{Daubechies}). We see again how learning the smoothness of the data from the AF can substantially improve the estimation. The time domain ``smoothing filter'' is now less sparse but has an intrinsic width in $t$, {\em cf} Figure \ref{fig9}(b). 

The next signal is a portion of a recording from a Pipistrellus bat
\footnote{see www.londonbats.org.uk}. The full recording is quite long -- over  59,417 time points, and we focus at a section towards the end of the recording. Fig. \ref{fig2}(c) shows the sparsity of the signal in the ambiguity domain, which would not be well approximated if limited to a small box near the origin. The two estimates of the covariance are given in Fig \ref{fig3} and \ref{fig4} where highly localized events in time are detected. The quantile/quantile- (qq-) plot (\ref{fig5}(c)) shows two populations, and $\widehat{\theta}=0.21,$ which is quite high. The raw Wigner plot is very noisy, and we prefer to here show the Rihazceck distribution of the signal (see Figure \ref{fig7}(a) and (b)), corresponding to a raw Fourier transform of the time-shifted moments with no additional windowing for better representation. The spectrogram loses precision, 
see Figure \ref{fig7}(c). The actual smoothing filter is quite sparse, but informative, see Figure \ref{fig9}(c), highlighting the cyclostationary nature of the signal.

\begin{figure}[!htbp]
\begin{center}
\begin{minipage}[]{0.36\textwidth}
\centering
\includegraphics[width=\textwidth]{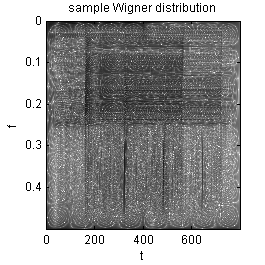}
(a)
\end{minipage}
\begin{minipage}[]{0.36\textwidth}
\centering
\includegraphics[width=\textwidth]{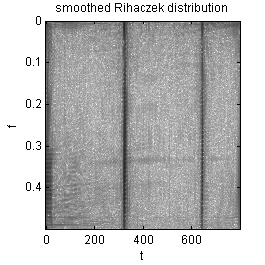}
(b)
\end{minipage}
\\
\begin{minipage}[]{0.36\textwidth}
\centering
\includegraphics[width=\textwidth]{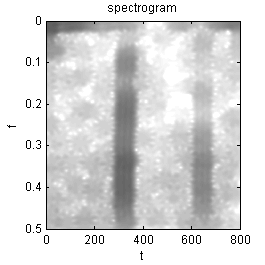}
(c)
\end{minipage}
\begin{minipage}[]{0.36\textwidth}
\centering
\includegraphics[width=\textwidth]{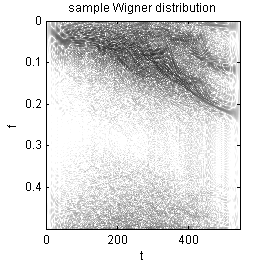}
(d)
\end{minipage}
\begin{minipage}[]{0.36\textwidth}
\centering
\includegraphics[width=\textwidth]{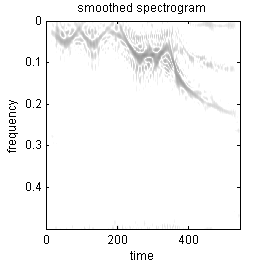}
(e)
\end{minipage}
\begin{minipage}[]{0.36\textwidth}
\centering
\includegraphics[width=\textwidth]{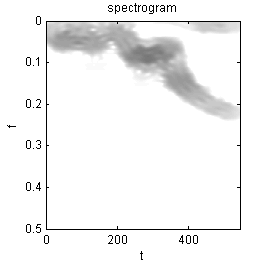}
(f)
\end{minipage}
\end{center}
\caption{Time-frequency representations of some of the sample signals. Subplot (a) is the raw Wigner distribution on a dB scale of the pipistrellus signal, (b) is the spectrogram of the moments estimated from Ambiguity domain and (c) is the raw spectrogram from the data. Subplot (d) is the raw Wigner distribution on a dB scale of the oceanographic signal, (e) is the spectrogram of the moments estimated from Ambiguity domain and (f) is the raw spectrogram from the data. 
\label{fig7}}
\end{figure}
The final signal is from an application in oceanography, see \cite{Richardson1}.
We propose to analyze experimental data corresponding to the velocity from a Lagrangian drifter. There are more components present in the data than merely the signal from a vortex (see {\em e.g.} \cite{Lilly10}). 
We plot the  estimated AF and covariance in Figures \ref{fig2}(d) and \ref{fig3}(d). Again this is a highly oscillatory signal correlated over long time intervals. The qq-plot of the AF (Fig \ref{fig5}(c)) now almost breaks the 2-component model in favour of a 3-component model, and $\theta$ is relatively high, namely $0.31$.
The raw Wigner distribution \ref{fig7}(d) is very noisy, and it is hard to make out the presence of distinct components. We look at both the spectrogram and the smoothed version of the Wigner distribution (again smoothing with Hermite windows) (Figs \ref{fig7}(f) and
\ref{fig7}(e)) and see that whilst the Wigner distribution has some issues with interference in the representation (at about $f\approx 0.11$), the localization of the instantaneous frequency path is much improved compared to the spectrogram. These paths are used to define summaries of the data (see {\em e.g.} \cite{Lilly10}). The temporal smoothing filter uses much of the local structure to improve estimation.

\section{Discussion}
The aim of this paper was to introduce a new class of nonstationary signals, and propose appropriate inference procedures for their second order properties. This was done by introducing ambiguity sparse processes, assumed to consist of a collection of high intensity regions in the ambiguity domain. Inference was implemented directly in the ambiguity domain, and estimators were transferred back into the time domain, so that they could be corrected into valid estimators of second order structure, an action that is {\em not} equivalent to requiring the time-frequency distribution to be real and positive. In general for an arbitrary harmonizable process the sinusoids will not diagonalize the covariance matrix of the observed sample, and so positivity of the ``spectrum'' {\em cannot} be a requirement to specify a valid process.
The estimators we propose of second order structure can however still be represented in terms of a time-varying spectrum, using any choice of such an object that is suitable for the process in question, once a valid covariance has been estimated. In practice we would first estimate the time-varying moments, correct them into valid estimators, and then subsequently choose from possible time-varying spectral representations. In our experience the correction has little effect on the pictorial impression of most time-varying spectra, as this de-facto seems to add a positive constant to the time-varying spectrum. Most summaries in the time-frequency plane compare {\em relative} magnitude, making this uniform level change of little interest.

In statistics interest has focused on various procedures using the spectrogram, or variants thereof ({\em e.g.} \cite{Ombao2001} or \cite{Dahlhaus1997}),
and this has a number of desirable properties such as positivity of the  spectrum, but also a number of shortcomings, see \cite{cohen}, mainly in terms of what processes can be treated well by such a framework. We have not solved the (intractable) problem of producing a time-varying spectral representation for an arbitrary nonstationary process, but nor is it reasonable to expect to do so.

A common complaint in statistics is that computing time-varying spectra easily degenerates into making ``pretty pictures'', and is not really an important inference problem in its own right. However in many practical problems summaries of time series are defined directly from their time-varying spectrum, we mention in particular oceanography \cite{Lilly2010b}, the analysis of biomedical signals, {\em e.g.} \cite{Unser,DeVille,Cranstoun}, and various branches of physics. To be able to obtain good summaries, the estimate of the time-frequency representation must itself be sufficiently good and the representation must be chosen appropriately, not smoothing out important features of interest. A great weakness of  existing methods is the necessity of assuming many smoothness properties of the covariance of the signal before  estimation. The difference between white noise and an interesting signal is exactly structure and concentration: and allowing a more wide range of possible structure can help us detect otherwise less easily characterized behavior. For stationary time series estimation was enabled because we assumed that the sinusoids corresponded to the eigenvectors of the covariance matrix, and so the positivity of the spectrum was required for a process to be valid. We reiterate that for an arbitrary harmonizable process there is no reason why this should ``nearly'' be the case.

Comparing with other methods that use sparsity to estimate time-varying spectra, such as \cite{Flandrin2010,Godsill}, is that we assume sparsity in the ambiguity domain rather than in the time-frequency domain. Furthermore we have provided a stochastic model for the types of signals where this inference method will be appropriate. By thresholding or correcting the eigenvalues of the eigendecomposition of the estimated covariance matrix, our estimated matrix {\em is} guaranteed to be a valid covariance matrix, another problem with many existing methods.
Excessive model flexibility is a curse; we cannot anticipate that it is possible to estimate any arbitrary nonstationary process. We relax the straight jacket of excessive smoothness so that we can at least estimate a larger class of nonstationary processes to encompass more realistic data sequences. In many applications as the sampling rate is matched to the actual bandwidth of the observed phenomenon, excessive smoothing will never recover the phenomena of interest, and more of the structure of the signal needs to be modelled and utilized, as otherwise several important features of the data will be smoothed out. 
\section*{Acknowledgement}
The author gratefully acknowledges the financial support from the EPSRC via EP/I005250/1 and many enlightening discussions with Dr Heidi Hindberg, Norut, about nonstationary processes.

\setcounter{section}{0}
\renewcommand{\thesection}{{\Alph{section}}}
\renewcommand{\thesubsection}{\Alph{section}.\arabic{subsection}}
\renewcommand{\thesubsubsection}{\Alph{section}.\arabic{subsection}.\arabic{subsubsection}}
\renewcommand{\theequation}{\Alph{section}-\arabic{equation}}
\section{Proof of Theorem 1 \label{AppendixA}}
\subsection{Expectation of the EMAF}
We start by noting that (see \cite[p.2384, eqn.~4]{Hindberg2009}), but adjusted to a sampling period not necessarily set to one:
\begin{alignat}{1}
\E\left\{\widehat{A}(\nu,\tau]\right\}&=
\int_{0}^{\frac{1}{2\Delta t}}\int_{-f}^{\frac{1}{2\Delta t}-f} e^{i\pi(N-1+\tau)\Delta t(\nu'-\nu)}D_{N-|\tau|}\left(\Delta t (\nu'-\nu)\right) e^{2i\pi f \tau \Delta t}  S(\nu',f)\;d\nu'\,df,
\end{alignat}
with the usual definition of the scaled Dirichlet kernel of \cite[p.~102]{PercivalWalden1993}:
\begin{equation}
D_{N}\left(\nu\right)=\frac{\sin(\pi N \nu)}{\sin(\pi \nu)}.
\end{equation}
By definition the dual frequency spectrum can be written in terms of the AF as:
\[S(\nu,f)=\Delta t\sum_{\tau=-\infty}^{\infty} A_{\tau}(\nu) e^{-2i\pi f \tau\Delta t}.\]
We then have that:
\begin{alignat}{1}
\label{eq:eAF}
\E\left\{\widehat{A}(\nu,\tau]\right\}&=
\int_{-\frac{1}{2\Delta t}}^{\frac{1}{2\Delta t}}\Delta t\sum_{\tau'=-\infty}^{\infty}A_{\tau'}(\nu')
D_{N-|\tau|}\left(\Delta t (\nu'-\nu)\right) e^{i\pi(N-1+\tau)\Delta t(\nu'-\nu)}e_{\tau'-\tau}(\nu') \;d\nu',
\end{alignat}
where the sequence $e_{\tau}$ is now defined by:
\begin{eqnarray}
\nonumber
e_{\tau}(\nu)&=&\int_{\max(-\nu,0)}^{\min(\frac{1}{2\Delta t}-\nu,\frac{1}{2\Delta t})}e^{-2i\pi f \tau \Delta t}\,df\\
&=&\left(\frac{1}{2\Delta t}-|\nu|\right) e^{-i\pi\tau\Delta t(\frac{1}{2\Delta t}-\nu)}{\mathrm{sinc}}\left(\pi\tau\Delta t(\frac{1}{2\Delta t}-|\nu|) \right),
\label{eqn:dirichlet}
\end{eqnarray}
with (as usual) ${\mathrm{sinc}}\left(f\right)=\sin(f)/f$. Note that $ e_{\tau}(\nu)$ therefore has an implicit and homogeneous dependence on $\Delta t$. 
We can write the expectation of the EMAF in terms of the AF, see \eqref{eq:eAF}. 
We see that the theoretical support of the ambiguity function, given by the points where $|A_{\tau'}(\nu')|$ is non-negligible is ``smeared out'' when sampled at a fixed $\Delta t$ and $N$ by the two kernel functions $D_{N-\tau}\left(\nu'-\nu\right)$ and $e_{\tau'-\tau}(\nu')$. The effects of this convolution must be further investigated when modelling the structure of the AF.
The value of Eqn \eqref{eq:eAF} sampled at unit sampling intervals for popular models like stationary processes, or uniformly modulated processes, is normally ${\cal O}(N-|\tau|)$, see \cite{Hindberg2009}, but depends on the model for $A_\tau(\nu)$. 
{\tiny \begin{alignat}{1}
\nonumber
\E\left\{\widehat{A}_\tau(\nu)\right\}&=\int_{-\frac{1}{2\Delta t}}^{\frac{1}{2\Delta t}}\Delta t\sum_{\tau'=-\infty}^{\infty}A_{\tau'}(\nu')
D_{N-|\tau|}\left(\Delta t (\nu'-\nu)\right) e^{i\pi(N-1+\tau)\Delta t(\nu'-\nu)}e_{\tau'-\tau}(\nu') \;d\nu'\\
\label{eq:exp1}
&\overset{(1)}{=}\int_{-\frac{1}{2\Delta t}}^{\frac{1}{2\Delta t}}\int_{-\infty}^{\infty} A_{u' /\Delta t}(\nu')
D_{N-|\tau|}\left(\Delta t (\nu'-\nu)\right) e^{i\pi(N-1+\tau)\Delta t(\nu'-\nu)}e_{(u'-u)/\Delta t}(\nu') \;d\nu'\,du'+{\cal O}\left(\Delta t|A_\tau(\nu)| \right).\end{alignat}}
If we are far from the $K$ singularities then we find we can Taylor expand the smooth function $A_{u' /\Delta t}(\nu')$ in a Taylor series
\begin{equation}
A_{u' /\Delta t}(\nu')=A_{u /\Delta t}(\nu)+\bigtriangledown A_{u /\Delta t}(\nu) \cdot \begin{pmatrix} \nu-\nu'\\
(u-u')/\Delta t
\end{pmatrix}+\dots .
\end{equation}
We take $\nu'=\nu+\xi_1/T$ (to catch the contributions of the Dirichlet kernel) and $u'=u+\xi_2\Delta t$, where $T(\tau)=(N-|\tau|)\Delta t$ (suppressing the $\tau$ dependence of this variable as appropriate).
\\
For large $N$ and small $\Delta t$ we therefore find
\begin{eqnarray}
\nonumber
D_{N-|\tau|}\left(\Delta t (\nu'-\nu)\right)&=&\frac{\sin(\pi \xi_1)}{\pi \xi_1/N}+{\cal O}\left(\frac{1}{N} \right)\\
\nonumber
e^{i\pi(N-1+\tau)\Delta t(\nu'-\nu)}e_{(u'-u)/\Delta t}(\nu')&=&e^{i\pi(N-1+\tau)\frac{\xi_1}{N}}\left(\frac{1}{2\Delta t}-|\nu+\xi_1/T|\right) e^{-i\pi(u'-u)(\frac{1}{2\Delta t}-\nu-\xi_1/T)}\\
\nonumber
&&{\mathrm{sinc}}\left(\pi(u'-u) (\frac{1}{2\Delta t}-|\nu+\xi_1/T|) \right)\\
\nonumber
&=&e^{i\pi(N-1+\tau)\frac{\xi_1}{N}} e^{-i\pi\xi_2 (\frac{1}{2}-\nu\Delta t-\xi_1/N)}\frac{\sin\left(\pi\xi_2  (\frac{1}{2}-|\nu\Delta t+\xi_1/N|) \right)}{\pi \xi_2 \Delta t}.
\end{eqnarray}
Therefore the integrals become
\begin{eqnarray}
\nonumber
\E\left\{\widehat{A}_\tau(\nu)\right\}&=&\int_{-\infty}^{\infty}\int_{-\infty}^{\infty} \left(A_{u /\Delta t}(\nu)+\bigtriangledown A_{u /\Delta t}(\nu) \cdot \begin{pmatrix} \xi_1/N\\
\xi_2 \Delta t
\end{pmatrix}+\dots \right)\\
\nonumber
&&
\left(\frac{\sin(\pi \xi_1)}{\pi \xi_1/N}+{\cal O}\left(\frac{1}{N} \right) \right)\left[e^{i\pi(N-1+\tau)\frac{\xi_1}{N}} e^{-i\pi\xi_2 (\frac{1}{2}-\nu\Delta t)}\left(\frac{\sin(\pi \xi_2(\frac{1}{2}-\nu|\Delta t|))}{\pi \xi_2\Delta t} \right)+\dots\right] \;\frac{d\xi_1}{T}\,d\xi_2\Delta t+\dots\\
\nonumber
&=&A_{\tau}(\nu)\int_{-\infty}^{\infty}\int_{-\infty}^{\infty}\frac{\sin(\pi \xi_1)}{\pi \xi_1}e^{i\pi \xi_1} e^{-i\pi\xi_2 (\frac{1}{2}-\nu\Delta t)}\frac{\sin(\pi \xi_2(\frac{1}{2}-\nu|\Delta t|))}{\pi \xi_2 \Delta t} d\xi_1 d\xi_2+\dots\\
&=&\frac{1}{\Delta t}A_{\tau}(\nu)\int_{-\infty}^{\infty}\int_{-\infty}^{\infty} I(\omega_1<1)I(\omega_2\in[\max(-\nu \Delta t,0),\min(\frac{1}{2}-\nu \Delta t,\frac{1}{2})])d^2{\bm{\omega}} \nonumber
\\
&=&A_{\tau}(\nu) \times 1 \times \left(\frac{1}{2\Delta t}-|\nu| \right)+\dots .
\label{first_exp}
\end{eqnarray}
Thus the expectation of this estimator far away from the singularities depends both on the sampling rate $\Delta t$ and how close we are to the Nyquist frequency.
For the points close to the singularities we instead use the model
and the concentration of the AF in order to derive its expectation. We use
\begin{equation}
A_{\tau}^{(k)}(\nu)= \frac{{\cal B}^{(k)}(\nu,\tau/N)}
{\left[\Delta t^2\left(\nu-\nu_0^{(k)}\right)^2+\frac{\left(\tau-\tau_0^{(k)}\right)^2}{N^2}\right]^{\delta^{(k)}}},
\label{AFmodelb}
\end{equation}
where ${\cal B}^{(k)}(\nu,u)$ is a twice differentiable function. We write  $(\nu_0^{(k)},\tau_0^{(k)})=
(\nu_0^{(k)},u_0^{(k)}N)$.
We can now rewrite the integral using a change of variables of 
\[\nu=\nu_0^{(k)}+\frac{\lambda_1}{T(\tau)},\quad \nu'=\nu_0^{(k)}+\frac{\xi_1}{T(\tau)},\quad \lambda_2=\tau-\tau_0^{(k)}=\Delta t (u-u_0^{(k)}),\quad \xi_2=\tau'-\tau_0^{(k)},\]
therefore ending up with (for some suitable choice of $1>\gamma>0$)
\begin{alignat}{1}
\E\left\{\widehat{A}_\tau(\nu)\right\}&\overset{(2)}{=} \sum_k\left(\frac{1}{2\Delta t}-|\nu|\right)\int_{-N^{\gamma}}^{N^{\gamma}}
\int_{-(N-|\tau|)^{\gamma}}^{(N-|\tau|)^{\gamma}}\frac{{\cal B}^{(k)}\left(\nu+\frac{\xi_1-\lambda_1}{T(\tau)},\frac{\tau}{N}+\frac{\xi_2-\lambda_2}{N}\right)}
{\left[\left(\frac{\xi_1}{N-|\tau|}\right)^2+\left(\frac{\xi_2}{N}\right)^2\right]^{\delta^{(k)}}}
D_{N-|\tau|}\left(\Delta t \left[\frac{\xi_1-\lambda_1}{T(\tau)}\right]\right) \nonumber\\
\nonumber
& e^{i\pi(N-1+\tau)\Delta t\left[\frac{\xi_1-\lambda_1}{T(\tau)}\right]}e^{-i\pi(\xi_2-\lambda_2)\Delta t(\frac{1}{2\Delta t}-(\nu+\frac{\xi_1-\lambda_1}{T}))
}\frac{\sin\left(\pi(\xi_2-\lambda_2)\Delta t(\frac{1}{2\Delta t}-|\nu+\frac{\xi_1-\lambda_1}{T(\tau)}|) \right)}{\pi(\xi_2-\lambda_2)\left(\frac{1}{2}-|\nu\Delta t|\right)}
\\
&\label{eq:exp12}
\frac{d\xi_1}{(N-|\tau|)\Delta t}\,d\xi_2\Delta t+{\cal O}(N^{2\delta^{(k)}-1} ).\end{alignat}
A term is added in \eqref{eq:exp1}(1) because there is an error due to the Riemann approximation to the sum. Subsequently (step \eqref{eq:exp12}(2)) there is a change of variables.  Let $\widehat{A}_\tau(\nu)=\sum_k\widehat{A}_\tau^{(k)}(\nu)$, separating the components from each ``island'' of contribution.
Thus the component renormalized ambiguity is, ignoring terms because we take only the first parts of a Taylor series of ${\cal B}^{(k)}(\nu,\tau/N)$, as well as the Dirichlet kernel:
\begin{alignat}{1}
\E\left\{\frac{\widehat{A}_\tau^{(k)}(\nu)}{{A}_\tau^{(k)}(\nu)}\right\}
\nonumber
&\approx \left(\frac{1}{2\Delta t}-|\nu|\right)
\int_{-N^{\gamma}}^{N^{\gamma}}\int_{-(N-|\tau|)^{\gamma}}^{(N-|\tau|)^{\gamma}}
\frac{\left[\lambda_1^2+\lambda_2^2 \right]^{\delta^{(k)}}}
{\left[\xi_1^2+\left(\frac{N-|\tau|}{N}\right)^2\xi_2^2\right]^{\delta^{(k)}}}
\frac{\sin\left(\pi(\xi_1-\lambda_1)\right)}{\pi (\xi_1-\lambda_1)} e^{i\pi (\xi_1-\lambda_1)}\\
\nonumber
& e^{-i\pi(\xi_2-\lambda_2)(\frac{1}{2}-\Delta t\nu)
}\frac{\sin\left(\pi(\xi_2-\lambda_2)(\frac{1}{2}-|\Delta t\nu|) \right)}{\pi(\xi_2-\lambda_2)(\frac{1}{2}-|\Delta t\nu|)}
d\xi_1d\xi_2.
\end{alignat}
Note that 
changing the limits in the integral from $-N^{\gamma}$ to $-\infty$ makes no appreciable difference as the contributions that have been added behave like:
\begin{eqnarray}
I&\sim& \int \int r^{-2\delta}\frac{1}{r\cos(\theta)-\lambda_1}\frac{1}{r\sin(\theta)-\lambda_1}\,d\theta \,r dr\sim \int_{N^{\gamma}}^N r^{-2\delta-1}\;dr=\frac{-1}{2\delta}\left[r^{-2\delta}\right]^{N}_{N^{\gamma}}\sim N^{-2\delta \gamma},
\end{eqnarray}
thus we need $\gamma>0$. We also need to convince ourselves that the integral is really ${\cal O}(1)$ as claimed, and let
{\tiny \begin{eqnarray*}
{\cal J}_1(\delta;\bm{\lambda})&=&(\frac{1}{2\Delta t}-|\nu|)\int_{-\infty}^{\infty}\int_{-\infty}^{\infty}
\frac{\left[\lambda_1^2+\lambda_2^2\right]^{\delta^{(k)}}}
{\left[\xi_1^2+\xi_2^2\right]^{\delta^{(k)}}}
\frac{\sin\left(\pi(\xi_1-\lambda_1)\right)}{\pi (\xi_1-\lambda_1)} e^{i\pi (\xi_1-\lambda_1)}e^{-i\pi(\xi_2-\lambda_2)(\frac{1}{2}-\nu\Delta t)
}\frac{\sin\left(\pi(\xi_2-\lambda_2)(\frac{1}{2}-|\Delta t\nu|) \right)}{\pi(\xi_2-\lambda_2)(\frac{1}{2}-|\Delta t\nu|)(\frac{1}{2}-|\Delta t\nu|)}
d\xi_1 d\xi_2.
\end{eqnarray*}}
It follows ${\cal J}_1$ is convergent if $0<\delta^{(k)}<\frac{1}{2},$ which can easily be shown directly by splitting the range of integration, and bounding each component (using the long-range decay), by implementing the integration in one variable after the other and then using the long range decay. The special case of $\delta^{(k)}=\frac{1}{2}$ is also fine, and this can be seen from using the fact that the 2-D Fourier transform\footnote{We need to define the Fourier transform carefully as $\xi_1$ has been a frequency variable and $\xi_2$ a time variable. The FT therefore has the opposite sign in the complex exponential for the two variables.}. The FT with canonical variable $\bm{\omega}$ of 
\[\frac{\sin\left(\pi \xi_1\right)}{\pi \xi_1} e^{i\pi \xi_1}e^{-i\pi \xi_2(\frac{1}{2}-\nu\Delta t)
}\frac{\sin\left(\pi\xi_2(\frac{1}{2}-|\Delta t\nu|) \right)}{\pi\xi_2},\] is a band-pass filter\footnote{You would expect $\omega_1$ to range in $[-1/2,1/2]$ but because we observe $t$ in $[0,T]$ this would not be the case.} $I(\omega_1 \in [0,1])$ and $I(\omega_2 \in [\max(-\Delta t\nu,0),\min(\frac{1}{2}-\Delta t\nu,\frac{1}{2})])$. We therefore see that the renormalized variables $\xi_1$ has a canonical variable restricted so that the FT of $\nu$ is restricted to $[0,T]$ which is exactly the interval we have observed the data over (and as the data is nonstationary we cannot go beyond that interval). The second variable $\xi_2$ is restricted in frequency to a range that is limited because of the act of making the signal analytic.
Using \cite[6.561.14]{gradshteyn} with $\omega=\sqrt{\omega_1^2+\omega_2^2}$
as the radius of the Fourier variable and requiring\footnote{We numerically approximate finite length sums using infinite domain integrals, however this makes the FT periodic in $\omega_1$.} $\delta^{(k)}\in\left(\frac{1}{4},\frac{3}{4}\right)$
{\small\begin{eqnarray}
\nonumber
{\cal F}\left\{\left[\xi_1^2+\xi_2^2\right]^{-\delta^{(k)}} \right\}&=&Q(\bm{\omega})=
\int_{-\infty}^{\infty} \int_{-\infty}^{\infty} \xi^{-2\delta^{(k)}}e^{-2i \pi \bm{\xi}^T\bm{\omega}}\,d^2\bm{\xi}
=2\pi \int_0^{\infty} \xi^{-2\delta^{(k)}}J_0(2\pi \xi \omega)\xi d\xi
\\
&=&
2\pi \times\left\{
\begin{array}{lcr}
2^{1-2\delta^{(k)}}
\left(2\pi \omega\right)^{-2+2\delta^{(k)}}\frac{\Gamma\left(1-\delta^{(k)} \right)}{\Gamma\left(\delta^{(k)}\right)}& {\mathrm{if}} & \frac{1}{4}<\delta^{(k)}<\frac{1}{2}\\
\left(2\pi \omega\right)^{-1}& {\mathrm{if}} & \delta^{(k)}=\frac{1}{2}\\
2^{1-2\delta^{(k)}}
\left(2\pi \omega\right)^{-2+2\delta^{(k)}}\frac{\Gamma\left(1-\delta^{(k)} \right)}{\Gamma\left(\delta^{(k)}\right)}& {\mathrm{if}} & \frac{1}{2}<\delta^{(k)}<1
\end{array}\right.\,{\mathrm{if}}\,\begin{array}{l}\omega_1\in[0,\frac{1}{2}]\\
\omega_2\in[0,\frac{1}{2}]
\end{array}
\end{eqnarray}}
Thus using Plancherel's theorem and the change of variable $2\pi \bm{\omega}=\tilde{\bm{\omega}}$
\footnote{We now get an extra factor of 2 from the periodicity in $\omega_1$.}
\begin{eqnarray}
\nonumber
{\cal J}_1(\delta;\bm{\lambda})&=&\frac{2\pi}{\Delta t} \left[\lambda_1^2+\lambda_2^2\right]^{\delta^{(k)}}
\int_{0}^{2\pi}\int_{2\pi\max(0,-\Delta t \nu)}^{2\pi\min(\frac{1}{2},\frac{1}{2}-|\nu|\Delta t)}2^{1-2\delta^{(k)}}
Q(\widetilde{\bm{\omega}})d^2\tilde{\bm{\omega}}\\
\nonumber
&\sim &\frac{4\pi}{\Delta t}  \left[\lambda_1^2+\lambda_2^2\right]^{\delta^{(k)}}
\int_{0}^{\infty}\int_{0}^{\infty}2^{1-2\delta^{(k)}}
\tilde{\omega}^{-2+2\delta^{(k)}}\frac{\Gamma\left(1-\delta^{(k)} \right)}{\Gamma\left(\delta^{(k)}\right)}e^{i \tilde{\bm{\omega}}^T\bm{\lambda}}\,d^2\tilde{\bm{\omega}}\\
\nonumber
&= & \frac{1}{\Delta t}\left[\lambda_1^2+\lambda_2^2\right]^{\delta^{(k)}}
\int_{0}^{\infty}2^{1-2\delta^{(k)}}
 \omega^{-1+2\delta^{(k)}}\frac{\Gamma\left(1-\delta^{(k)} \right)}
{\Gamma\left(\delta^{(k)}\right)}\frac{1}{2}J_0( \lambda \omega)\,d\omega
=\frac{1}{2\Delta t},
\end{eqnarray}
if $\delta^{(k)}\in [1/4,3/4]$. This should be compared to Eqn \eqref{first_exp}, and we see that if the AF is sufficiently concentrated we no longer have a loss of information proportional to $\frac{1}{2\Delta t}-|\nu|$.  As we are only interested in the second order structure up to proportionality, the multiplication by $\frac{1}{2\Delta t}$ does not concern us, and could easily be corrected for.
The truncation in the canonical variable of $\xi_1$ exactly implies that the time variation is (for not normalized variables) truncated onto $[0,T].$ The second truncation is in the frequency variable $f$ (once renormalized), which just ensures if $\nu>0$ that $0<\omega_2<(1/2-|\Delta t \nu|)$.

In theory we also wish to include stationary processes, as well as uniformly modulated processes in the class of investigated processes. A sample from an analytic stationary process with autocovariance sequence $\widetilde{M}_\tau$ has an Empirical Ambiguity function with expectation \cite{Hindberg2009}: 
\begin{equation}
\mu_\tau(\nu)=D_{N-|\tau|}(\Delta t \nu)\widetilde{M}_\tau e^{-i\pi \nu\Delta t(N+\tau-1)},
\end{equation}
(compared with the infinite sample $A_\tau(\nu)=\delta(\nu)\widetilde{M}_\tau$)
whilst a sample from the analytic signal of a uniformly modulated process with the Fourier transform of the analytic extension of the modulation function corresponding to $\Sigma(\nu)$ has expectation:
\begin{equation}
\mu_\tau(\nu)=\left(\frac{1}{2\Delta t}-|\nu| \right)\Sigma(\nu){\mathrm{sinc}}
\left[\left(\frac{1}{2\Delta t}-|\nu| \right)\tau\Delta t\right]e^{-i\pi \nu\Delta t(N+\tau-1)}+{\cal O}(1),
\end{equation}
compare with the infinite length sample $A_\tau(\nu)=\delta_\tau \Sigma(\nu)$.
Both these means decay away from the point $(0,0)$, if with different decay in $\nu$ and $\tau$. These can be approximated by the ellipsoidal decay model (see Eqn (8) in the paper).

\subsection{Variance of the EMAF}
An expression for the variance of a general harmonizable process is given in \cite[Eqn A-4]{Hindberg2009}, adjusted to the case of $\Delta t \neq 1$ namely:
\begin{alignat}{1}
\nonumber
\var\left\{\widehat{A}_{\tau}(\nu)\right\}+\left|\E\left\{\widehat{A}_{\tau}(\nu)\right\}\right|^2&=
\int_{0}^{\frac{1}{2\Delta t}} \int_{0}^{\frac{1}{2\Delta t}} \int_{0}^{\frac{1}{2\Delta t}} \int_{0}^{\frac{1}{2\Delta t}} \E\left\{dZ (f_1)dZ^\ast (f_2)
dZ^{\ast}(\alpha_1)dZ(\alpha_2)\right\}\\
\nonumber
&e^{i2\pi (f_1-\alpha_1)\tau \Delta t}e^{i\pi(f_1-f_2-\alpha_1+\alpha_2)(N-1-\tau)\Delta t}D_{N-|\tau|}(\Delta t(f_1-f_2-\nu))
\\ &D_{N-|\tau|}(\Delta t(\alpha_1-\alpha_2-\nu)).
\label{ta-da}
\end{alignat}
We shall simplify this expression somewhat.
We have  that with
\[f_1=\nu'+f,\;f_2=f,\; \alpha_1=\nu''+f',\quad \alpha_2=f'. \]
\begin{alignat}{1}
\nonumber
\var\left\{\widehat{A}_{\tau}(\nu)\right\}+\left|\E\left\{\widehat{A}_{\tau}(\nu)\right\}\right|^2&=
\int_{0}^{\frac{1}{2\Delta t}}\int_{-f}^{\frac{1}{2\Delta t}-f} 
\int_{0}^{\frac{1}{2\Delta t}}\int_{-f'}^{\frac{1}{2\Delta t}-f'} 
e^{i\pi(N-1+\tau)\Delta t(\nu'-\nu)}D_{N-|\tau|}\left(\Delta t(\nu'-\nu)\right) \\
\nonumber
&  e^{2i\pi f \tau \Delta t}\E\left\{ d{ Z}(\nu'+f)d{ Z}^*(f) d{ Z}^*(\nu''+f')d{ Z}(f')\right\}\\
&
e^{-i\pi(N-1+\tau)\Delta t(\nu''-\nu)}D_{N-|\tau|}\left(\Delta t (\nu''-\nu)\right) e^{-2i\pi f' \tau \Delta t}.
\label{eq:eafdef}
\end{alignat}
We assume that $\{d{ Z}(f)\}$ is a Gaussian process and then use Isserlis' theorem to obtain that
\begin{alignat*}{1}
\var\left\{\widehat{A}_{\tau}(\nu)\right\}&=
\int_{0}^{\frac{1}{2\Delta t}}\int_{-f}^{\frac{1}{2\Delta t}-f} 
\int_{0}^{\frac{1}{2\Delta t}}\int_{-f'}^{\frac{1}{2\Delta t}-f'} e^{i\pi(N-1+\tau)\Delta t(\nu'-\nu)}D_{N-|\tau|}
\left(\Delta t (\nu'-\nu)\right) e^{2i\pi f \tau \Delta t}
\\
&  \left(\E\left\{ d{ Z}(\nu'+f)d{ Z}^*(\nu''+f')\right\}\E\left\{ d{ Z}^*(f)d{ Z}(f')\right\}\right.\\
&\left.+\E\left\{ d{ Z}(\nu'+f)d{ Z}(f')\right\}\E\left\{ d{ Z}^*(f)d{ Z}^*(\nu''+f')\right\}\right)\\
&
e^{-i\pi(N-1+\tau)\Delta t(\nu''-\nu)}D_{N-|\tau|}\left(\Delta t(\nu''-\nu)\right) e^{-2i\pi f' \tau\Delta t} .
\end{alignat*}
Given we have assumed propriety of $Z(t)$ we obtain that $\E\left\{d{ Z}(\nu'+f)d{ Z}(f')\right\}\E\left\{ d{ Z}^*(f)d{ Z}^*(\nu''+f')\right\}=0$, and so the remaining term gives us:
\begin{alignat*}{1}
\var\left\{\widehat{A}_{\tau}(\nu)\right\}&=
\int_{0}^{\frac{1}{2\Delta t}}\int_{-f}^{\frac{1}{2\Delta t}-f} 
\int_{0}^{\frac{1}{2\Delta t}}\int_{-f'}^{\frac{1}{2\Delta t}-f'} 
e^{i\pi(N-1+\tau)\Delta t(\nu'-\nu)}D_{N-|\tau|}\left(\Delta t(\nu'-\nu)\right) e^{2i\pi f \tau \Delta t}\\
&  \E\left\{ d{ Z}(\nu'+f)d{ Z}^*(\nu''+f')\right\}\E\left\{ d{ Z}^*(f)d{ Z}(f')\right\}\\
&
e^{-i\pi(N-1+\tau)\Delta t(\nu''-\nu)}D_{N-|\tau|}\left(\Delta t(\nu''-\nu)\right) e^{-2i\pi f' \tau \Delta t} \\
&=\int_{0}^{\frac{1}{2\Delta t}}\int_{-f}^{\frac{1}{2\Delta t}-f} 
\int_{0}^{\frac{1}{2\Delta t}}\int_{-f'}^{\frac{1}{2\Delta t}-f'} 
e^{i\pi(N-1+\tau)\Delta t(\nu'-\nu)}D_{N-|\tau|}\left(\Delta t(\nu'-\nu)\right) e^{2i\pi f \tau\Delta t}\\
&  S(\nu'+f-\nu''-f',\nu''+f')S^*(f-f',f')\\
&
e^{-i\pi(N-1+\tau)\Delta t(\nu''-\nu)}D_{N-|\tau|}\left(\Delta t(\nu''-\nu)\right) e^{-2i\pi f' \tau\Delta t} \;d\nu'\,df\;d\nu'\,df'.
\end{alignat*}
We can re-write this in terms of the ambiguity function and redefine $\nu'$ as well as $\nu''$ to have
\begin{eqnarray*}
\var\left\{\widehat{A}_{\tau}(\nu)\right\}&=&\int_{0}^{\frac{1}{2\Delta t}}\int_{-f-\nu}^{\frac{1}{2\Delta t}-f-\nu} 
\int_{0}^{\frac{1}{2\Delta t}}\int_{-f'-\nu}^{\frac{1}{2\Delta t}-f'-\nu} 
e^{i\pi(N-1+\tau)\Delta t\nu'}D_{N-|\tau|}\left(\Delta t\nu'\right) e^{2i\pi f \tau \Delta t}\\
&&  \times\Delta t\sum_{\tau'=-\infty}^{\infty} A_{\tau'}(\nu'+f-\nu''-f')e^{-2i\pi \tau'\Delta t(\nu''+f'+\nu)}\\
&&  \times\Delta t\sum_{\tau''=-\infty}^{\infty} A_{\tau''}^*(f-f')e^{2i\pi \tau''\Delta tf'}\\
&&
e^{-i\pi(N-1+\tau)\Delta t\nu''}D_{N-|\tau|}\left(\Delta t \nu''\right) e^{-2i\pi f' \tau \Delta t} \;d\nu'\,df\;d\nu''\,df'.
\end{eqnarray*}   
We need to put in our knowledge for when $A_{\tau'}(\nu'+f-\nu''-f')A_{\tau''}^*(f-f')$ is large to simplify this expression. We start by a change of variables
\[\omega=\nu''-\nu',\quad \nu'=s,\quad f-f'=u,\quad f=b.\]
With this change of variables we have that (with the limits $\{l_p\}$ depending on the outer integral dummy variables, as well as $\nu$, and are in general complicated objects) with $T=T(\tau)=\Delta t (N-|\tau|)$:
\begin{eqnarray}
\nonumber
\var\left\{\widehat{A}_{\tau}(\nu)\right\}&=&
\int_{s=l_3}^{l_4}\int_{b=l_1}^{l_2} \int_{\omega=l_7}^{l_8}\int_{u=l_5}^{l_6}\Delta t\sum_{\tau'=-\infty}^{\infty}  \Delta t\sum_{\tau''=-\infty}^{\infty}
\\
\nonumber
&& e^{i\pi(N-1+\tau)s\Delta t}D_{N-|\tau|}\left(\Delta t s\right) e^{2i\pi b \tau \Delta t}\\
\nonumber
&&   A_{\tau'}(u-\omega)e^{-2i\pi \tau'\Delta t(s+\omega-u+b+\nu)}  A_{\tau''}^*(u)e^{2i\pi \tau''\Delta t(-u+b)}\\
\nonumber
&&
e^{-i\pi(N-1+\tau)\Delta t(s+\omega)}D_{N-|\tau|}\left(\Delta t(s+\omega)\right) e^{-2i\pi (-u+b)\tau \Delta t} \;db\,ds\;du\,d\omega.\end{eqnarray}
We implement yet several other changes of variables
\[u-\omega=\nu_0^{(l)}+\frac{\xi_l}{T},\quad  u=\nu_0^{(k)}+\frac{\nu_k}{T},\]
to focus on the locations close to the finite number of singularities
\begin{eqnarray}
\nonumber
\var\left\{\widehat{A}_{\tau}(\nu)\right\}
&=&\sum_k \sum_l\int_{s=l_3}^{l_4}\int_{b=l_1}^{l_2} \int_{\xi_k}\int_{\eta_k}\Delta t\sum_{\tau'=-\infty}^{\infty}  \Delta t\sum_{\tau''=-\infty}^{\infty}
\\
\nonumber
&& e^{i\pi(N-1+\tau)s\Delta t}D_{N-|\tau|}\left(\Delta t s\right) e^{2i\pi b \tau \Delta t}\\
\nonumber
&&   A_{\tau'}(\nu_0^{(l)}+\xi_l/T)e^{-2i\pi \tau'\Delta t(s-\nu_0^{(l)}-\xi_l/T+b+\nu)}  A_{\tau''}^*(\nu_0^{(k)}+\eta_k/T)e^{2i\pi \tau''\Delta t(-(\nu_0^{(k)}+\eta_k/T)+b)}\\
\nonumber
&&
e^{-i\pi(N-1+\tau)\Delta t(s+(\nu_0^{(k)}+\eta_k/T)-\nu_0^{(l)}-\xi_l/T)}D_{N-|\tau|}\left(\Delta t(s+(\nu_0^{(k)}+\eta_k/T)-\nu_0^{(l)}-\xi_l/T))\right)\\
&& e^{-2i\pi (-\nu_0^{(k)}-\eta_k/T+b)\tau \Delta t}\;ds \;db\,\frac{d\xi_k}{T}\,\frac{d\eta_k}{T}
+{\cal O}(1).\end{eqnarray}
As of yet we have not used the model, or made any approximations. 
We see from this expression that we only need to worry about $k=l$ (as otherwise contributions become negligible) and so we obtain for some $\gamma>0$
\begin{eqnarray*}
\nonumber
\var\left\{\widehat{A}_{\tau}(\nu)\right\}&=&
\sum_{k=0}^K \int_{s=l_3}^{l_4}\int_{b=l_1}^{l_2} \Delta t\sum_{\tau'=-\infty}^{\infty}  \Delta t\sum_{\tau''=-\infty}^{\infty}
 D_{N-|\tau|}\left(\Delta t s\right)   e^{-2i\pi \tau'\Delta t(s-\nu_0^{(l)}-\frac{\xi_k}{T}+b+\nu)} e^{2i\pi \tau''\Delta t(-(\nu_0^{(k)}+\eta_k/T)+b)}\\
\nonumber
&&
e^{-i\pi(N-1+\tau)\Delta t(\eta_k/T-\xi_k/T)} e^{2i\pi (\nu_0^{(k)}+\frac{\eta_k}{T})\tau \Delta t}\;ds \;db
\int_{T^{\gamma}l_7}^{T^{\gamma}l_8}\int_{T^{\gamma}l_5}^{T^{\gamma}l_6} A_{\tau'}(\nu_0^{(l)}+\xi_l/T)
 A_{\tau''}^*(\nu_0^{(k)}+\eta_k/T)\\
 &&D_{N-|\tau|}\left(\Delta t(s+\eta_k/T-\xi_l/T))\right)\,\frac{d\xi_k}{T}\,\frac{d\eta_k}{T}
 +{\cal O}(1).\end{eqnarray*}
 We implement a final change of variable 
of
\[\alpha=T(\tau) s,\] 
and find
\begin{eqnarray}
\nonumber
\var\{\widehat{A}_\tau(\nu) \}
 &=&\sum_k\int_{b=0}^{\frac{1}{2\Delta t}-\nu} \int_{\alpha=N^\gamma(-b-\nu)\Delta t}^{N^\gamma(\frac{1}{2\Delta t}-b-\nu)\Delta t} \sum_{\tau'=-\infty}^{\infty}  \sum_{\tau''=-\infty}^{\infty}
\\
\nonumber
&& D_{N-|\tau|}\left(\frac{\alpha}{N-|\tau|}\right)   e^{-2i\pi \tau'\Delta t(-\nu_0^{(k)}+b+\nu+\frac{\alpha-\xi_k}{T})} e^{2i\pi \tau''\Delta t(-\nu_0^{(k)}+b+\frac{\eta_k}{T})}\\
\nonumber
&&
e^{-i\pi(N-1+\tau)\Delta t(\eta_k/T-\xi_k/T)} e^{2i\pi (\nu_0^{(k)}+\frac{\eta_k}{T})\tau \Delta t}\;\frac{d\alpha}{(N-|\tau|)\Delta t} \;db
\\
&&\int_{T^{\gamma}l_7}^{T^{\gamma} l_8}\int_{T^{\gamma}l_5}^{T^{\gamma}l_6} A_{\tau'}(\nu_0^{(k)}+\xi_k/T(\tau))\nonumber
 A_{\tau''}^*(\nu_0^{(k)}+\eta_k/T(\tau))\\
 &&D_{N-|\tau|}\left(\frac{\alpha+\eta_k-\xi_k}{N-|\tau|}\right)\,\frac{d\xi_k}{N-|\tau|}\,\frac{d\eta_k}{N-|\tau|}+{\cal O}(1).
\end{eqnarray}
Then for a suitable $\gamma>0$ after implementing the integral in the variable $b$ and restricting the range of integration so that the limit of the integrals increasing in $N$ is non-negligible:
{\small \begin{eqnarray}
\nonumber
\var\{\widehat{A}_\tau(\nu) \}&=&\sum_k \int_{-\infty}^{\infty} \sum_{\tau'=-\infty}^{\infty}  \sum_{\tau''=-\infty}^{\infty}\left(\frac{1}{2\Delta t }-|\nu| \right)e^{-i\pi (\tau'-\tau'')\Delta t \left(\frac{1}{2\Delta t }-|\nu| \right)}{\mathrm{sinc}}
\left(\pi (\tau'-\tau'')\left(\frac{1}{2 }-|\nu|\Delta t  \right) \right)\\
\nonumber
&& (N-|\tau|)^2\frac{\sin(\pi\alpha)}{\pi\alpha}   e^{-2i\pi \tau'\Delta t(-\nu_0^{(k)}+\nu)} e^{-2i\pi \tau''\Delta t\nu_0^{(k)}}
e^{-i\pi(N-1+\tau)\Delta t(\eta_k/T-\xi_k/T)} e^{2i\pi \nu_0^{(k)}\tau \Delta t}\;\frac{d\alpha}{(N-|\tau|)\Delta t} 
\\
&&\int_{T^{\gamma}l_7}^{T^{\gamma} l_8}\int_{T^{\gamma}l_5}^{T^{\gamma}l_6} A_{\tau'}(\nu_0^{(k)}+\xi_k/T(\tau))
 A_{\tau''}^*(\nu_0^{(k)}+\eta_k/T(\tau))e^{2i\pi \Delta t(\tau'\frac{\alpha-\xi_k}{T}-\tau''\frac{\eta_k}{T}+\tau \frac{\eta_k}{T})}\nonumber\\
 &&\frac{\sin(\pi(\alpha+\eta_k-\xi_k))}{\pi(\alpha+\eta_k-\xi_k)}\,\frac{d\xi_k}{N-|\tau|}\,\frac{d\eta_k}{N-|\tau|}\nonumber
+{\cal O}(1)\\\nonumber
 &=&\sum_k \int_{-\infty}^{\infty} \sum_{\tau'=-\infty}^{\infty}  \sum_{\tau''=-\infty}^{\infty}\left(\frac{1}{2\Delta t }-|\nu| \right)e^{-i\pi (\tau'-\tau'')\Delta t \left(\frac{1}{2\Delta t }-|\nu| \right)}{\mathrm{sinc}}
\left(\pi (\tau'-\tau'')\left(\frac{1}{2 }-|\nu|\Delta t  \right) \right)\\
\nonumber
&& \frac{\sin(\pi\alpha)}{\pi\alpha}   e^{-2i\pi \tau'\Delta t\nu}
e^{-i\pi(\eta_k-\xi_k)} e^{2i\pi \nu_0^{(k)}(\tau-\tau''+\tau') \Delta t}\;\frac{d\alpha}{(N-|\tau|)\Delta t} \int_{T^{\gamma}l_7}^{T^{\gamma} l_8}\int_{T^{\gamma}l_5}^{T^{\gamma}l_6} A_{\tau'}(\nu_0^{(k)}+\xi_k/T(\tau))
\\
\nonumber
&&A_{\tau''}^*(\nu_0^{(k)}+\eta_k/T(\tau))e^{2i\pi \Delta t(\tau'\frac{\alpha-\xi_k}{T}-\tau''\frac{\eta_k}{T}+\tau \frac{\eta_k}{T})}\frac{\sin(\pi(\alpha+\eta_k-\xi_k))}
{\pi(\alpha+\eta_k-\xi_k)}\,d\xi_k\,d\eta_k+{\cal O}(1)\\
&=&\sum_k\frac{1}{\Delta t}\left|{\cal B}^{(k)}(\nu_0^{(k)},\tau_0^{(k)})\right|^2 (N-|\tau|)^{4\delta^{(k)} -1}{\cal J}_2^{(k)}(\delta^{(k)},\nu,\tau)+{\cal O}(1),
\end{eqnarray}}
this {\em defining} ${\cal J}_2^{(k)}(\delta^{(k)},\nu,\tau)$
where we have implemented an integral in $\alpha$ using
{\small\begin{eqnarray}\int_{-\infty}^{\infty} 
\nonumber
\frac{\sin(\pi\alpha)}{\pi\alpha}
\frac{\sin(\pi(\alpha+\eta_k-\xi_k))}{\pi(\alpha+\eta_k-\xi_k)}e^{2i\pi\alpha\frac{\tau'}{N-|\tau|}}d\alpha
&=&\int_{\max\left(-\frac{1}{2},\frac{\tau'}{N-|\tau|}\right)}^{\min\left(\frac{1}{2},\frac{\tau'}{N-|\tau|}+\frac{1}{2}\right)} e^{-2i\pi t(\eta_k-\xi_k)}\,dt
\\
&=&e^{-i\pi \frac{\tau'}{N-|\tau|}(\eta_k-\xi_k)}\frac{\sin\left(\pi\left((\eta_k-\xi_k)\left(1+\frac{\tau'}{N-|\tau|} \right)\right) \right)}{\pi (\eta_k-\xi_k)}.
\end{eqnarray}}
We then get with $\tilde{\tau}_{1k}=\tau'-\tau_0^{(k)}$ as well as $\tilde{\tau}_{2k}=\tau''-\tau_0^{(k)}$
\begin{eqnarray}
\nonumber
{\cal J}_2^{(k)}(\delta^{(k)},\nu,\tau)&=&\left(\frac{1}{2\Delta t }-|\nu| \right)
\int_{-\infty}^{\infty}\int_{-\infty}^{\infty} \int_{-\infty}^{\infty}  
\int_{-\infty}^{\infty}\frac{1}{\left[\xi_k^2+\tilde\tau^2_{1k} \right]^{\delta^{(k)}}}
 \frac{1}{\left[\eta_k^2+\tilde\tau^2_{2k} \right]^{\delta^{(k)}}}e^{-i\pi (\tilde{\tau}_{1k}-\tilde{\tau}_{2k})\Delta t \left(\frac{1}{2\Delta t }-|\nu| \right)}\\
\nonumber
&& {\mathrm{sinc}}
\left(\pi (\tilde{\tau}_{1k}-\tilde{\tau}_{2k})\left(\frac{1}{2 }-|\nu|\Delta t  \right) \right)   e^{-2i\pi (\tilde{\tau}_{1k}+\tau_k^{(0)})\Delta t\nu}
 e^{2i\pi \nu_0^{(k)}(\tau+\tilde{\tau}_{1k}-\tilde{\tau}_{2k}) \Delta t}e^{2i \pi \frac{\Delta t}{T}\left(-\tau'\xi_k-\tau''\eta_k+\tau \eta_k  \right)}\\
&& e^{-i\pi(1+\frac{\tau'}{N-|\tau|})(\eta_k-\xi_k)}\frac{\sin\left(\pi\left((\eta_k-\xi_k)\left(1+\frac{\tau'}{N-|\tau|} \right)\right) \right)}{\pi (\eta_k-\xi_k)}
\,d\xi_k\,d\eta_k d\tilde{\tau}_{1k} d\tilde{\tau}_{2k}+o(1).
\end{eqnarray}
We define the projection operator (assuming $\delta^{(k)}\in\left[\frac{1}{4},1\right]$):
{\small\begin{eqnarray}
\nonumber
{\cal P}g_1(\xi_k,\widetilde{\tau}_1)&=&\frac{1}{\Delta t}\left(\frac{1}{2 }-|\nu|\Delta t \right)\int_{-\infty}^{\infty}\int_{-\infty}^{\infty} e^{-i\pi(\eta_k-\xi_k)}\frac{\sin\left(\pi\left(\eta_k-\xi_k\right) \right)}{\pi (\eta_k-\xi_k)} \frac{1}{\left[\eta_k^2+\tilde\tau^2_{2k} \right]^{\delta^{(k)}}}e^{-i\pi (\tilde{\tau}_{1k}-\tilde{\tau}_{2k})\Delta t \left(\frac{1}{2\Delta t }-|\nu| -2\nu^{(0)}_k\right)}\\
&&{\mathrm{sinc}}
\left(\pi (\tilde{\tau}_{1k}-\tilde{\tau}_{2k})\left(\frac{1}{2 }-|\nu|\Delta t  \right) \right)\,d\eta_k \,d\tilde{\tau}_{2k},\\
{\cal F}{\cal P}g_1(\bm{\omega})&=& \frac{1}{\Delta t}I(\omega_1 \in [0,1])
2^{1-2\delta}(2\pi \omega)^{-2+2\delta^{(k)}}\frac{\Gamma(1-\delta^{(k)})}{\Gamma(\delta^{(k)})}
I(\omega_2-\nu^{(0)}_k\Delta t\in [\max(-\nu\Delta t ,0),\min(\frac{1}{2 }-\nu\Delta t ],\frac{1}{2})).
\nonumber\end{eqnarray}}
Finally we note 
{\small \begin{eqnarray}
{\cal J}_2^{(k)}(\delta^{(k)},\nu,\tau)&=&\frac{1}{\Delta t}e^{-2i\pi \Delta t(\tau_0^{(k)}\nu- \nu_0^{(k)}\tau)}\int_{-\infty}^{\infty} \int_{-\infty}^{\infty} \frac{1}{\left[\xi_k^2+\tilde\tau^2_{1k} \right]^{\delta^{(k)}}}e^{-2i\pi \tilde\tau_{1k}\nu\Delta t}{\cal P}g_1(\xi_k,\widetilde{\tau}_1)
\,d\xi_k d\tilde{\tau}_1+o(1).
\end{eqnarray}}
We have $\bm{\omega}(\nu)=\left[\omega_1+\nu\Delta t,\omega_2\right]$ (again assuming $\delta^{(k)}\in\left[\frac{1}{4},1\right]$)
\begin{eqnarray}
{\cal F}\left\{\frac{1}{\left[\xi_k^2+\tilde\tau^2_{1k} \right]^{\delta^{(k)}}}e^{-2i\pi \tilde\tau_{1k}\nu} \right\}(\bm{\omega})&=&
2^{1-2\delta}(2\pi \omega(\nu \Delta t))^{-2+2\delta^{(k)}}\frac{\Gamma(1-\delta^{(k)})}{\Gamma(\delta^{(k)})}.
\end{eqnarray}
Using Plancherel theorem we have 
{\small \begin{eqnarray}
{\cal J}_2^{(k)}(\delta^{(k)},\nu,\tau)&=&\frac{1}{\Delta t}e^{-2i\pi \Delta t (\tau_0^{(k)}\nu- \nu_0^{(k)}\tau)}\int_{-\infty}^{\infty}  \int_{-\infty}^{\infty}  2^{1-2\delta}(2\pi \omega(\nu\Delta t))^{-2+2\delta^{(k)}}\frac{\Gamma(1-\delta^{(k)})}{\Gamma(\delta^{(k)})}\\
&&I(\omega_1 \in [0,1])2^{1-2\delta}(2\pi \omega)^{-2+2\delta^{(k)}}\frac{\Gamma(1-\delta^{(k)})}{\Gamma(\delta^{(k)})}
I(\omega_2-\Delta t\nu^{(0)}_k\in [\max(-\nu\Delta t,0),\min(\frac{1}{2 }-|\nu|\Delta t ],\frac{1}{2}))d^2\bm{\omega}\nonumber\\
&&+o(1)
\nonumber .
\end{eqnarray}}
Again the limits of these truncations are set like in the expectation limits.
Truncating the integrals using the indicator functions we find that:
\begin{eqnarray}
{\cal J}_2^{(k)}(\delta^{(k)},\nu,\tau)&\sim& \frac{1}{\Delta t}2^{2-4\delta}
\left[\frac{\Gamma(1-\delta^{(k)})}{\Gamma(\delta^{(k)})}\right]^2\int_{0}^{1}  \int_{\Delta t\max(\nu_k^{(0)}-\nu,0)}^{\min(\frac{1}{2}+(\nu_k^{(0)}-|\nu|)\Delta t ,\frac{1}{2}+\Delta t\nu_k^{(0)})} 
\omega(\Delta t\nu)^{-2+2\delta^{(k)}} \omega^{-2+2\delta^{(k)}}
\frac{d^2\bm{\omega}}{(2\pi)^2}
\nonumber 
\end{eqnarray}
From this it is possible to deduce that ${\cal J}_2^{(k)}(\delta^{(k)},\nu,\tau)$ is finite for the previously mentioned range of $\frac{1}{4}\le \delta^{(k)}<\frac{3}{4}$. Furthermore the latter integral has the second range depending on $\frac{1}{2 }-|\nu| \Delta t$, and so
%
 it follows:
\begin{eqnarray}
\var\{\widehat{A}_\tau(\nu) \}  &=&{\cal O}\left(\frac{1}{\Delta t}\left(\frac{1}{2  \Delta t}-|\nu|  \right) \sum_k (N-|\tau|)^{4\delta^{(k)} -1}  \right).
\end{eqnarray}
The finiteness of the limiting integrals therefore follow from exactly the same calculations as in the case of the mean.
This should be compared in units with our expression for the expectation square. We see that the units agree, but that the expectation will be large compared to the variance, and should therefore be recoverable. We have if $(\nu,\tau) \in {\cal D}_k$ with $\delta=\max \delta_k$
\begin{equation}
\frac{\left|\E\{\widehat{A}_\tau(\nu) \}\right|^2}{\var\{\widehat{A}_\tau(\nu) \}}={\cal O}\left(\frac{\frac{1}{4\Delta t^2}(N-|\tau|)^{4\delta^{(k)}}
}{\left(\frac{1}{2 \Delta t}-|\nu| \right) (N-|\tau|)^{4\delta -1}}\right)={\cal O}\left((N-|\tau|)^{1+4(\delta^{(k)}-\delta)}\left(\frac{1}{2 }-|\nu|  \Delta t\right)^{-1}\right),
\end{equation}
which naturally is independent of the units of the problem, and becomes larger with increasing $N$ and decreasing $\Delta t$. If $\delta^{(k)}=\delta$ for all $k$ then this normalized expectation does not depend on individual $\delta^{(k)}$ but only needs $\delta>0$.

\subsection{Relation of the EMAF}
Finally to determine all properties of the ambiguity function we also need to determine its relation sequence.
We find by direct calculation that (extending in \cite[Eqn A-7]{Hindberg2009} to $\Delta t \neq 1$):
\begin{alignat}{1}
\nonumber
\rel\left\{\widehat{A}_{\tau}(\nu)\right\}+\E^2\left\{\widehat{A}_{\tau}(\nu)\right\}&=
\int_{0}^{\frac{1}{2\Delta t}}\int_{-f}^{\frac{1}{2\Delta t}-f} 
\int_{0}^{\frac{1}{2\Delta t}}\int_{-f'}^{\frac{1}{2\Delta t}-f'} 
e^{i\pi(N-1+\tau)\Delta t(\nu'-\nu)}D_{N-|\tau|}\left(\Delta t(\nu'-\nu)\right) \\
\nonumber
&  e^{2i\pi f \tau\Delta t}\E\left\{ d{ Z}(\nu'+f)d{ Z}^*(f) d{ Z}(\nu''+f')d{ Z}^*(f')\right\}\\
&
e^{i\pi(N-1+\tau)\Delta t(\nu''-\nu)}D_{N-|\tau|}\left(\Delta t(\nu''-\nu)\right) e^{2i\pi f' \tau \Delta t} \;d\nu'\,df\;d\nu''\,df'.
\label{eq:relafdef}
\end{alignat}
We have assume that $d{ Z}(f)$ is a Gaussian process and then use Isserliss theorem to obtain that:
\begin{alignat*}{1}
\rel\left\{\widehat{A}_{\tau}(\nu)\right\}&=
\int_{f=0}^{\frac{1}{2\Delta t}}\int_{\nu'=-f}^{\frac{1}{2\Delta t}-f} 
\int_{f'=0}^{\frac{1}{2\Delta t}}\int_{\nu''=-f'}^{\frac{1}{2\Delta t}-f'} e^{i\pi(N-1+\tau)\Delta t(\nu'-\nu)}D_{N-|\tau|}
\left(\Delta t(\nu'-\nu)\right) e^{2i\pi f \tau\Delta t}
\\
&  \left(\E\left\{ d{ Z}(\nu'+f)d{ Z}(\nu''+f')\right\}\E\left\{ d{ Z}^*(f)d{ Z}^*(f')\right\}\right.\\
&\left.+\E\left\{ d{ Z}(\nu'+f)d{ Z}^*(f')\right\}\E\left\{ d{ Z}^*(f)d{ Z}(\nu''+f')\right\}\right)\\
&
e^{i\pi(N-1+\tau)\Delta t(\nu''-\nu)}D_{N-|\tau|}\left(\Delta t(\nu''-\nu)\right) e^{2i\pi f' \tau\Delta t} \;d\nu'\,df\;d\nu'\,df'.
\end{alignat*}
Again, with the assumption that $dZ(f)$ is proper we are left with
\begin{alignat*}{1}
\rel\left\{\widehat{A}_{\tau}(\nu)\right\}&=
\int_{f=0}^{\frac{1}{2\Delta t}}\int_{\nu'=-f}^{\frac{1}{2\Delta t}-f} 
\int_{f'=0}^{\frac{1}{2\Delta t}}\int_{\nu''=-f'}^{\frac{1}{2\Delta t}-f'} e^{i\pi(N-1+\tau)\Delta t(\nu'-\nu)}D_{N-|\tau|}
\left(\Delta t(\nu'-\nu)\right) e^{2i\pi f \tau\Delta t}
\\
&\E\left\{ d{ Z}(\nu'+f)d{ Z}^*(f')\right\}\E\left\{ d{ Z}^*(f)d{ Z}(\nu''+f')\right\}\\
&
e^{i\pi(N-1+\tau)\Delta t(\nu''-\nu)}D_{N-|\tau|}\left(\Delta t(\nu''-\nu)\right) e^{2i\pi f' \tau\Delta t} \;d\nu'\,df\;d\nu'\,df'\\
&=\int_{0}^{\frac{1}{2\Delta t}}\int_{-f}^{\frac{1}{2\Delta t}-f} 
\int_{0}^{\frac{1}{2\Delta t}}\int_{-f'}^{\frac{1}{2\Delta t}-f'} e^{i\pi(N-1+\tau)\Delta t(\nu'+\nu''-2\nu)}D_{N-|\tau|}
\left(\Delta t (\nu'-\nu)\right) e^{2i\pi (f+f') \tau\Delta t}
\\
&S(\nu'+f-f',f')S(\nu''+f'-f,f)D_{N-|\tau|}\left(\Delta t(\nu''-\nu)\right)  \;d\nu'\,df\;d\nu'\,df'.
\end{alignat*}
We only expect to get contributions to the dual-frequency spectrum if
$\nu'\approx -f+f'+\nu_0^{(k)}$ and $\nu''\approx f-f'+\nu_0^{(l)}$ for some $k$ and $l$. Furthermore the Dirichlet kernels will only contribute unless
$\nu' \approx \nu$ and $\nu''\approx \nu$. If $\nu=0$ then this is not a problem, and the ambiguity function is real. 
If $\nu \neq \nu_0^{(k)}$ for some $k$ then the Dirichlet kernels will have arguments like $\pm(f-f')+\nu_0^{(l)}-\nu,$ and so if one contribution is large, the other is small. To bound contributions the integrals are done explicitly in the ambiguity domain, and it can be shown that for the empty points the contributions are $o(N-|\tau|)$. For a strictly underspread process, it is possible to show they are ${\cal O}\left(\log(N)\right)$, see \cite{Hindberg2009}.

\section{Correlation of EMAF for White Noise}\label{distofEMAF}
We find that the correlation for white noise is given with $\mu_{\tau}(\nu)=\E\left\{
\widehat{A}_{\tau}(\nu) \right\}$ by (see Eqn \eqref{ta-da})
{\small \begin{eqnarray}
\nonumber
&&\cov\left\{\widehat{A}_{\tau_1}(\nu_1), \widehat{A}_{\tau_2}(\nu_2)\right\}+\mu_{\tau_1}(\nu_1)\mu_{\tau_2}^\ast(\nu_2)
=\int_0^{\frac{1}{2\Delta t}} \int_{-f}^{\frac{1}{2\Delta t}-f} \int_0^{\frac{1}{2\Delta t}} \int_{-f_2}^{\frac{1}{2\Delta t}-f_2} e^{i\pi (\nu'-\nu_1)\Delta t(N-|\tau_1|-1)}e^{i2\pi f \tau_1\Delta t}
\\
\nonumber
&&D_{N-|\tau_1|}(\Delta t(\nu'-\nu_1))\E\left\{dZ(\nu'+f)dZ^\ast(f)dZ^\ast(\nu''+f_2)dZ(f_2)\right\}
e^{-i\pi (\nu''-\nu_2)\Delta t (N-|\tau_2|-1)}e^{-i2\pi f_2 \tau_2\Delta t}\\
&&D_{N-|\tau_2|}(\Delta t (\nu''-\nu_2)).
\end{eqnarray}}
Thus due to the propriety of $dZ(f)$ we have
\begin{eqnarray}
\nonumber
&&\cov\left\{\widehat{A}_{\tau_1}(\nu_1), \widehat{A}_{\tau_2}(\nu_2)\right\}=
\int_0^{\frac{1}{2\Delta t}} \int_{-f}^{\frac{1}{2\Delta t}-f} \int_0^{\frac{1}{2\Delta t}} \int_{-f_2}^{\frac{1}{2\Delta t}-f_2} \E\left\{dZ(\nu'+f)dZ^\ast(\nu''+f_2)\right\}
\\
 \nonumber
 && \E\left\{dZ(f_2)dZ^\ast(f)\right\}e^{i\pi (\nu'-\nu_1)\Delta t(N-|\tau_1|-1)}e^{i2\pi f \tau_1\Delta t }D_{N-|\tau_1|}(\Delta t (\nu'-\nu_1))
\\ && e^{-i\pi (\nu''-\nu_2)\Delta t (N-|\tau_2|-1)}e^{-i2\pi f_2 \tau_2\Delta t}D_{N-|\tau_2|}(\Delta t (\nu''-\nu_2)).\end{eqnarray}
Thus
{\small \begin{eqnarray*}
\nonumber
&&\cov\left\{\widehat{A}_{\tau_1}(\nu_1), \widehat{A}_{\tau_2}(\nu_2)\right\}
=\int_{-\frac{1}{2\Delta t}}^{\frac{1}{2\Delta t}} \int_{-\frac{1}{2\Delta t}}^{\frac{1}{2\Delta t}} \int_{\max(\nu',0)}^{\min\left(\frac{1}{2\Delta t}-\nu',\frac{1}{2\Delta t}\right)} \int_{\max(\nu'',0)}^{\min\left(\frac{1}{2\Delta t}-\nu'',\frac{1}{2\Delta t}\right)} \sigma^2\delta(\nu'+f-\nu''-f_2)\\
&&\sigma^2 \delta(f_2-f)e^{i\pi (\nu'-\nu_1)\Delta t(N-|\tau_1|-1)}e^{i2\pi f \tau_1\Delta t }D_{N-|\tau_1|}(\Delta t(\nu'-\nu_1))
 e^{-i\pi (\nu''-\nu_2)\Delta t (N-|\tau_2|-1)}e^{-i2\pi f_2 \tau_2\Delta t}\\
 && \nonumber D_{N-|\tau_2|}(\Delta  t (\nu''-\nu_2))\,df_1\,df_2d\nu'd\nu''.\end{eqnarray*}
 This simplifies to
{\small \begin{eqnarray*}
&&\cov\left\{\widehat{A}_{\tau_1}(\nu_1), \widehat{A}_{\tau_2}(\nu_2)\right\}
=\int_{-\frac{1}{2\Delta t}}^{\frac{1}{2\Delta t}} \int_{-\frac{1}{2\Delta t}}^{\frac{1}{2\Delta t}} \int_{\max(\nu',\nu'',0)}^{\min\left(\frac{1}{2\Delta t}-\nu',\frac{1}{2\Delta t}-\nu'',\frac{1}{2\Delta t}\right)} \sigma^4\delta(\nu'-\nu'') \\
 &&e^{i\pi (\nu'-\nu_1)\Delta t(N-|\tau_1|-1)}e^{i2\pi f \Delta t(\tau_1-\tau_2)}D_{N-|\tau_1|}(\Delta t(\nu'-\nu_1))
\\ && e^{-i\pi (\nu''-\nu_2)\Delta t(N-|\tau_2|-1)}D_{N-|\tau_2|}(\Delta t(\nu''-\nu_2))\,df\,d\nu'd\nu''\\
&=&\int_{-\frac{1}{2\Delta t}}^{\frac{1}{2\Delta t}} \int_{\max(\nu',0)}^{\min(\frac{1}{2\Delta t}-\nu',1/2)} \sigma^4
e^{i\pi (\nu'-\nu_1)\Delta t (N-|\tau_1|-1)}e^{i2\pi f (\tau_1-\tau_2)\Delta t}D_{N-|\tau_1|}(\Delta t(\nu'-\nu_1))
\\ && e^{-i\pi (\nu'-\nu_2)\Delta t (N-|\tau_2|-1)}D_{N-|\tau_2|}(\Delta t(\nu'-\nu_2))\,df\,d\nu'\\
&=&\sigma^4\int_{-\frac{1}{2\Delta t}}^{\frac{1}{2\Delta t}} \left[\frac{e^{i2\pi f \Delta t(\tau_1-\tau_2)}}{i2\pi  \Delta t(\tau_1-\tau_2)}\right]_{\max(\nu',0)}^{\min(\frac{1}{2\Delta t}-\nu',1/2)} 
e^{i\pi (\nu'-\nu_1)\Delta t (N-|\tau_1|-1)}D_{N-|\tau_1|}(\Delta t (\nu'-\nu_1))
\\ && e^{-i\pi (\nu'-\nu_2)\Delta t(N-|\tau_2|-1)}D_{N-|\tau_2|}(\Delta t(\nu'-\nu_2))\,df\,d\nu'
\end{eqnarray*}}
We now implement a change of variables of $\nu'=\xi/[\Delta t(N-\max(\tau_1,\tau_2))]$ and find with $\nu_1=j_1/(\Delta t(N-\max(\tau_1,\tau_2)))$ as well as $\nu_2=j_2/(\Delta t(N-\max(\tau_1,\tau_2)))$ :
{\small \begin{eqnarray*}
&&\cov\left\{\widehat{A}_{\tau_1}(\nu_1), \widehat{A}_{\tau_2}(\nu_2)\right\}
=\Delta t (N-\max(\tau_1,\tau_2))\sigma^4\left[\frac{e^{i2\pi f \Delta t(\tau_1-\tau_2)}}{i2\pi  (\tau_1-\tau_2)\Delta t}\right]_{\max(\nu_1,0)}^{\min(1/2\Delta t-\nu_1,1/2)} e^{i\pi (j_2-j_1)}\\
&&\int_{-\infty}^{\infty}\frac{\sin(\pi(\xi-j_1))\sin(\pi(\xi-j_2))}{\pi^2(\xi-j_1)(\xi-j_2)}\,d\xi+{\cal O}(1)\\
&=&(N-\max(\tau_1,\tau_2))\sigma^4\left[\frac{e^{i2\pi f \Delta t(\tau_1-\tau_2)}}{i2\pi  (\tau_1-\tau_2)}\right]_{\max(\nu_1,\nu_2,0)}^{\min(1/2\Delta t-\max(\nu_1,\nu_2),1/2)} e^{i\pi (j_2-j_1)}\delta_{j_1,j_2}+{\cal O}(1)\\
&=&(N-\max(\tau_1,\tau_2))\sigma^4\left[\frac{\exp(2\pi \Delta t(1/2\Delta t-\max(\nu_1,\nu_2)) i (\tau_1-\tau_2))-\exp(2\pi \max(\nu_1,\nu_2) i\Delta t (\tau_1-\tau_2))}{i2\pi  (\tau_1-\tau_2)}\right]\\
 \nonumber
 && e^{i\pi (j_2-j_1)}\delta_{j_1,j_2}
+{\cal O}(1) \\
&=&(N-\max(\tau_1,\tau_2))\sigma^4\left[\frac{(-1)^{j_1-j_2}\exp(-2\pi \Delta t\max(\nu_1,\nu_2)) i (\tau_1-\tau_2))-\exp(2\pi \Delta t \max(\nu_1,\nu_2) i (\tau_1-\tau_2))}{i2\pi  (\tau_1-\tau_2)}\right]\\
 && e^{i\pi (j_2-j_1)}\delta_{j_1,j_2}+{\cal O}(1).
\end{eqnarray*}}
This is zero depending on the frequencies of choice. If we fix $\nu_1$ then we have to pick times that agree with that choice of $\nu_1$ whilst if we fix $\tau_1$ we can make the process independent without worrying about the value of $\tau_1$ (apart from it being ${\cal O}(1)$), by just choosing the Fourier coefficients (this becomes more involved when $\tau$ is no longer order one). Interestingly, we can pick the $\tau$ values arbitrarily if we change the relative frequencies. 

\section{EBAYES Calculations}\label{EBAYEScalc}
These calculations resemble strongly those of \cite{Wang}.
We first marginalize the likelihood over the non-observed ${\cal B}_{\tau}(\nu)$. We find that starting from Eqn. (69) that
{\small \begin{eqnarray*}
f(\widehat{A}_\tau^{(N)}(\nu_k),\bpsi)&=&\int \int
\left\{\frac{(1-\rho)e^{-\frac{\left|\widehat{A}^{(N)}_{\tau}(\nu_k)\right|^2}{\overline{\cal V}}}}{\pi\overline{\cal V}} \delta\left({\cal B}_{\tau}(\nu_k)\right)+
\frac{\rho e^{-\frac{\left|\widehat{A}^{(N)}_{\tau}(\nu_k)-{\cal B}_{\tau}(\nu_k)\right|^2}{\overline{\cal V}}}}{\pi^2\overline{\cal V}\sigma_\cN^2} e^{-\frac{\left|{\cal B}_{\tau}(\nu_k)\right|^2}
{\sigma^2_\cN}} \right\}\,d{\cal B}_\tau (\nu_k)\\
&=&\frac{(1-\rho)}{\pi\overline{\cal V}} e^{-\frac{\left|\widehat{A}^{(N)}_{\tau}(\nu_k)\right|^2}{\overline{\cal V}}}
+\frac{\rho}{\pi(\overline{\cal V}+\sigma^2_{\cN})}e^{-\frac{\left|\widehat{A}^{(N)}_{\tau}(\nu_k)\right|^2}{\overline{\cal V}+\sigma^2_{\cN}}}.
\end{eqnarray*}}
This implies that our marginal likelihood for $\widehat{Q}_\tau(\nu_k)$ is in fact:
{\small \begin{eqnarray*}
f(\widehat{Q}_\tau(\nu_k),\bpsi)&=&2\frac{(1-\rho)}{\overline{\cal V}} \widehat{Q}_{\tau}(\nu_k)e^{-\frac{\left(\widehat{Q}_{\tau}(\nu_k)\right)^2}{\overline{\cal V}}}+2\frac{\rho}{\overline{\cal V}+\sigma^2_{\cN}}\widehat{Q}_{\tau}(\nu_k)
e^{-\frac{\left(\widehat{Q}_{\tau}(\nu_k)\right)^2}{\overline{\cal V}+\sigma^2_{\cN}}}
\end{eqnarray*}}
Thus:
\begin{eqnarray}
f\left(\widehat{Q}_{\tau}^2(\nu)
\right)&=&
\frac{(1-\rho)}{\overline{\cal V}} 
e^{-\frac{\widehat{Q}_{\tau}^2(\nu)}{\overline{\cal V}}}+\frac{\rho}{\overline{\cal V}+\sigma^2_{\cN}}
e^{-\frac{\widehat{Q}_{\tau}^2(\nu)}{\overline{\cal V}+\sigma^2_{\cN}}}
\end{eqnarray}

However if we just marginalize Eqn (69) over the phase then we have
{\small \begin{eqnarray}
\nonumber
f\left(\widehat{Q}_{\tau}^2(\nu), \cB_{\tau}(\nu)
\right)&=&\int_0^{2\pi}  f\left(\widehat{Q}_{\tau}^2(\nu)e^{-i\vartheta_{\tau}(\nu)},\cB_{\tau}(\nu)\right)
\widehat{Q}_{\tau}(\nu)\,d\vartheta_{\tau}(\nu)\\
\nonumber
&=&2\widehat{Q}_{\tau}(\nu)\frac{(1-\rho)}{\overline{\cal V}} 
e^{-\frac{\widehat{Q}^{2}_{\tau}(\nu_k)}{\overline{\cal V}}}\delta\left({\cal B}_{\tau}(\nu_k)\right)+
\frac{\rho}{\pi\overline{\cal V}\sigma_\cN^2}\widehat{Q}_{\tau}(\nu)\frac{1}{\pi} \int_0^{2\pi} e^{-\frac{\left|\widehat{A}^{(N)}_{\tau}(\nu_k)-{\cal B}_{\tau}(\nu_k)\right|^2}{\overline{\cal V}}}\\
&& e^{-\frac{\left|{\cal B}_{\tau}(\nu_k)\right|^2}
{\sigma^2_\cN}}\,d\vartheta_{\tau}(\nu).
\end{eqnarray}}
We define the posterior ratio to be
\begin{eqnarray}\rho^{(\mathrm{post})}_{\tau}(\nu_k)&=&
\frac{\frac{{\rho}}{(\overline{\cal V}+\sigma^2)} e^{-\frac{\widehat{Q}_{\tau}^{2}(\nu_k)}{\overline{\cal V}+\sigma^2}}}
{\frac{(1-\rho)}{\overline{\cal V}} e^{-\frac{\widehat{Q}_{\tau}^{2}(\nu_k)}{\overline{\cal V}}}+
\frac{\rho}{(\overline{\cal V}+\sigma^2)} e^{-\widehat{Q}_{\tau}^{2}(\nu_k)}}
\end{eqnarray}
Thus we may write
{\small \begin{eqnarray}
\nonumber
f\left(\cB_{\tau}(\nu)|\widehat{Q}_{\tau}(\nu_k)\right)
&=&\frac{\frac{(1-\rho)}{\overline{\cal V}} 
e^{-\frac{\left|\widehat{A}^{(N)}_{\tau}(\nu_k)\right|^2}{\overline{\cal V}}}\delta\left({\cal B}_{\tau}(\nu_k)\right)+
\frac{\rho}{\pi\overline{\cal V}\sigma_\cN^2}\frac{1}{2\pi} \int_0^{2\pi} e^{-\frac{\left|\widehat{A}^{(N)}_{\tau}(\nu_k)-{\cal B}_{\tau}(\nu_k)\right|^2}{\overline{\cal V}}} e^{-\frac{\left|{\cal B}_{\tau}(\nu_k)\right|^2}
{\sigma^2_\cN}}\,d\xi_{\tau}(\nu)}{\frac{(1-\rho)}{\overline{\cal V}} 
e^{-\frac{\left|\widehat{A}^{(N)}_{\tau}(\nu_k)\right|^2}{\overline{\cal V}}}+\frac{\rho}{\overline{\cal V}+\sigma^2_{\cN}}
e^{-\frac{\left|\widehat{A}^{(N)}_{\tau}(\nu_k)\right|^2}{\overline{\cal V}+\sigma^2_{\cN}}}}\\
&=&(1-\rho^{(\mathrm{post})}_{\tau}(\nu_k))\delta\left({\cal B}_{\tau}(\nu_k)\right)+\rho^{(\mathrm{post})}_{\tau}(\nu_k)
f_1(\cB_{\tau}(\nu)|\widehat{Q}_{\tau}(\nu)).
\end{eqnarray}}
With this definition we find that:
\begin{eqnarray}
f_1(\cB|\widehat{Q}_{\tau}(\nu))&=&\frac{(\overline{\cal V}+\sigma^2)}{\pi \overline{\cal V}\sigma^2} e^{\frac{\widehat{Q}_{\tau}^2(\nu)}{\overline{\cal V}+\sigma^2}}\frac{1}{2\pi}\int_0^{2\pi} e^{-\frac{\left|\widehat{A}-{\cal B}\right|^2}{\overline{\cal V}}} e^{-\frac{\left|{\cal B}\right|^2}
{\sigma^2_\cN}}\,d\vartheta.
\end{eqnarray}
Thus if we marginalize yet again over the angle then we find that:
\begin{eqnarray*}
f_1(\cQ;\hat{Q})&=&\frac{(\overline{\cal V}+\sigma^2)}{\overline{\cal V}\sigma^2}e^{\frac{\hat{Q}^2}{(\overline{\cal V}+\sigma^2)}}e^{-\frac{\cQ^2}{\sigma^2_N}}
\frac{\cQ}{2\pi^2}\int_0^{2\pi}\int_0^{2\pi} 
e^{-\frac{\hat{A}^2+\cQ^2-2\hat{Q}\cQ \cos(\xi+\vartheta)}{\overline{\cal V}}}\;d\xi\;d\vartheta\\
&=&\frac{(\overline{\cal V}+\sigma^2)}{\overline{\cal V}\sigma^2}e^{\frac{\hat{Q}^2}{(\overline{\cal V}+\sigma^2)}}e^{-\frac{\cQ^2}{\sigma^2}} 
\frac{\cQ}{2\pi^2}\int_0^{2\pi}\int_{0+\phi}^{2\pi+\phi} e^{-\frac{\hat{Q}^2+\cQ^2-2\hat{Q}\cQ\cos(\xi)}{\overline{\cal V}}}\;d\xi\;d\vartheta\\
&=&\frac{(\overline{\cal V}+\sigma^2)}{\overline{\cal V}\sigma^2}e^{\frac{\hat{Q}^2}{(\overline{\cal V}+\sigma^2)}}e^{-\frac{\cQ^2}{\sigma^2}} 
e^{-\frac{\hat{Q}^2+\cQ^2}{\overline{\cal V}}}
\frac{\cQ}{\pi}\int_0^{2\pi} 
e^{\frac{2\hat{Q}\cQ \cos(\xi)}{\overline{\cal V}}}\;d\xi.
\end{eqnarray*}
Thus we see that using \cite[3.339]{gradshteyn}:
\begin{eqnarray}
\nonumber
f_1(\cQ;\widehat{Q})&=&\frac{2}{\overline{\cal V}\lambda}
e^{\frac{\widehat{Q}^2}{(\overline{\cal V}+\sigma^2)}}e^{-\frac{\cQ^2}{\sigma^2}}
e^{-\frac{\widehat{Q}^2+\cQ^2}{\overline{\cal V}}}\cQ J_0\left(-i\frac{2\widehat{Q}
\cQ}{\overline{\cal V}}\right)\\
\nonumber
&=&\frac{2}{\overline{\cal V}\lambda}
e^{-\lambda\frac{\widehat{Q}^2}{ \overline{\cal V}}}
e^{-\frac{\cQ^2}{\lambda\overline{\cal V}}}\cQ
J_0\left(-i\frac{2\hat{Q}\cQ}{\overline{\cal V}}\right)\\
&\approx &N\left(\lambda \hat{Q},\frac{1}{2}\lambda\overline{\cal V}\right).
\end{eqnarray}

We then have
\begin{eqnarray}
F_1(\cQ|\hat{Q})&=& 2\int_0^{\cQ/\sqrt{\lambda\overline{\cal V}}}u
e^{-\left(u^2+\frac{\lambda \hat{Q}^2}{\overline{\cal V}}\right)}
J_0\left(-2i u\frac{\sqrt{\lambda}\hat{Q}}{\sqrt{\overline{\cal V}}} \right)\;d{u}.
\end{eqnarray}
We can check with Gradshteyn {\em et al.} \cite[6.614]{gradshteyn} that this CDF integrates to one, and recognize this is a Rice distribution with parameters $\sigma^2=\frac{1}{2}\lambda {\overline{\cal V}}$ and $\mu=\lambda \hat{Q}$. When the mean becomes large compared to the standard deviation, e.g. when
\[\frac{\sqrt{2\lambda} \hat{Q}}{\sqrt{\overline{\cal V}}}>>1 ,\]
then there can be no danger in using such an approximation.

\section{Variance of Estimator}
To invert the ambiguity function into a local moment sequence we calculate
{\small \begin{eqnarray}
\nonumber
\widehat{M}^{({\mathrm{eb}})}_\tau(t_n)&\approx&\frac{1}{2N\Delta t}
\sum_{k=-(N-1)}^{N-1}\widehat{A}^{({\mathrm{eb}})}_\tau(\nu_k)
e^{i2\pi \nu_k t_n},\quad {\nu}_k=\frac{k}{2N\Delta t}\\
\nonumber
\var\left\{\widehat{M}^{({\mathrm{eb}})}_\tau(t_n)\right\}&=&
\frac{1}{4N^2\Delta t^2}\sum_{k_1=-(N-1)}^{N-1}\sum_{k_2=-(N-1)}^{N-1}
e^{i2\pi (\nu_{k_1}-\nu_{k_2}) t_n}\cov\left\{ \widehat{A}^{({\mathrm{eb}})}_\tau(\nu_{k_1}),
\widehat{A}^{({\mathrm{eb}})}_\tau(\nu_{k_2})\right\}\\ \nonumber
&=&\frac{1}{4N^2\Delta t^2}\sum_{k_1=-(N-1)}^{N-1}\sum_{k_2=-(N-1)}^{N-1}
e^{i2\pi (\nu_{k_1}-\nu_{k_2}) t_n}\Theta_\tau(\nu_{k_1})\Theta_\tau^\ast(\nu_{k_2})\\
&& \nonumber \cov\left\{ \widehat{A}_\tau(\nu_{k_1}),
\widehat{A}_\tau(\nu_{k_2})\right\}\nonumber \\
\var\left\{\widehat{M}_\tau(t_n) \right\}&= &\frac{1}{4N^2\Delta t^2}\sum_{k_1=-(N-1)}^{N-1}
\sum_{k_2=-(N-1)}^{N-1}
e^{i2\pi (\nu_{k_1}-\nu_{k_2}) t_n}\cov\left\{ \widehat{A}_\tau(\nu_{k_1}),
\widehat{A}_\tau(\nu_{k_2})\right\}.
\label{bahgh}
\end{eqnarray}}
We see from appendix \ref{distofEMAF} that the correlation between frequencies for white noise with the same lag will become small if $\nu_{k_1}$ and $\nu_{k_2}$ are on the grid $\nu_k(\tau)=\frac{k}{N-\tau}$. The effect of being slightly off that grid (for small $\tau$) is negligible. Therefore if we consider this case 
\begin{eqnarray}
\nonumber
\var\left\{\widehat{M}^{({\mathrm{eb}})}_\tau(t_n)\right\}&\approx&
\frac{1}{N^2\Delta t^2}\sum_{k=-(N/2-1)}^{N/2-1}
\Theta_\tau^2(\nu_{k})\var\left\{ \widehat{A}_\tau(\nu_{2k_1})\right\}\\
\nonumber
&\sim &\frac{1}{T}\int_{-\frac{1}{2\Delta t}}^{\frac{1}{2\Delta t}} \Theta_\tau^2(\nu)\var\left\{ \widehat{A}_\tau(\nu)\right\}\;d\nu.
\end{eqnarray}
As $\Theta_\tau(\nu)=0$ with probability $\rho$, and $\rho \sim \frac{1}{N^{\gamma}}$ the variance will decrease with $\rho^2$ for white noise.
\end{document}